\documentclass[conference]{IEEEtran}

\usepackage{adjustbox}
\usepackage{amsfonts}
\usepackage{amsmath,amsthm}
\usepackage{algorithm}
\usepackage{booktabs}
\usepackage{caption}
\usepackage{cite}
\usepackage{comment}
\usepackage{enumitem}
\usepackage[T1]{fontenc}
\usepackage{filecontents}
\usepackage{graphicx}
\usepackage{hyperref}
\usepackage{cleveref}

\usepackage{longtable}
\usepackage{makecell}
\usepackage{multirow}
\usepackage{multicol}

\usepackage{pgfplots}
\pgfplotsset{compat=1.16}
\usepackage{pgfplotstable}
\usepackage{pifont}
\usepackage{stmaryrd}
\usepackage{subcaption}
\usepackage{scalerel}
\usepackage{tabu}
\usepackage{textcomp}
\usepackage{colortbl}
\usepackage{tikz}
\usepackage{thmtools}
\usepackage{thm-restate}

\usepackage[framemethod=tikz]{mdframed} 
\usetikzlibrary{shapes.geometric,arrows}
\usetikzlibrary{calc}
\usepackage{url}
\usepackage{xcolor}
\usepackage{xspace}

\usepackage{algpseudocode}

\newcommand{\lbl}{{\mvec{y}}}
\newcommand{\dataset}{\textbf{X}}
\newcommand{\ssample}{{\mvec{{x}}}}
\newcommand{\slbl}{\text{y}}

\newcommand{\usr}{\mathsf{C}}
\newcommand{\fm}{\mathcal{M}}

\newcommand{\ld}{{D}}
\newcommand{\fd}{\mathsf{D}}

\newcommand{\sss}[1]{\scriptscriptstyle #1}

\newcommand{\myparagraph}[1]{\noindent \textbf{#1.}}

\newcommand{\pitrain}{\ensuremath{\Pi_\mathsf{train}}}
\newcommand{\pistr}{\ensuremath{\Pi_\mathsf{train}}}

\newcommand{\valvec}{\ensuremath{\vb^\mathsf{val}}}

\newcommand{\Asoc}{\ensuremath{\mathcal{A}_\text{soc}}}
\newcommand{\Ap}{\ensuremath{\mathcal{A}_\text{p}}}
\newcommand{\Apsc}{\ensuremath{\mathcal{A}^{\text{p}}_\text{soc}}}
\newcommand{\srvr}{\mathcal{S}}

\newtheorem{theorem}{Theorem}

\newtheorem{definition}[theorem]{Definition}

\newcommand{\abs}[1]{| #1 |}

\newcommand{\Prob}{\ensuremath{\mathsf{Pr}}}

\newcommand{\Z}[1]{\ensuremath{\mathbb{Z}}_{2^{#1}}}

\newcommand{\Adv}{\ensuremath{\mathcal{A}}\xspace}
\newcommand{\Sim}{\ensuremath{\mathcal{S}}}

\newcommand{\share}[1]{\ensuremath{{\llbracket#1\rrbracket}}}

\makeatletter
\newsavebox{\@brx}
\newcommand{\llangle}[1][]{\savebox{\@brx}{\(\m@th{#1\langle}\)}%
	\mathopen{\copy\@brx\kern-0.6\wd\@brx\usebox{\@brx}}}
\newcommand{\rrangle}[1][]{\savebox{\@brx}{\(\m@th{#1\rangle}\)}%
	\mathclose{\copy\@brx\kern-0.6\wd\@brx\usebox{\@brx}}}
\makeatother

\newcommand{\add}{\ensuremath{\mathsf{add}}}
\newcommand{\mult}{\ensuremath{\mathsf{mult}}}

\newcommand{\Sh}{\ensuremath{\mathsf{sh}}}

\newcommand{{\piaSh}}[1]{\ensuremath{\Pi^{#1}_{\Sh}}}

\newcommand{\piSh}{\ensuremath{\Pi_{\Sh}}}

\newcommand{\FSh}{\ensuremath{\mathcal{F}_{\mathsf{shr}}}}

\newcommand{\FMul}{\ensuremath{\mathcal{F}_{\mathsf{mult}}}}
\newcommand{\FOp}{\ensuremath{\mathcal{F}_{\mathsf{op}}}}
\newcommand{\FPred}{\ensuremath{\mathcal{F}_{\mathsf{pred}}}}
\newcommand{\FmPred}{\ensuremath{\mathcal{F}_{\fm\mathsf{pred}}}}
\newcommand{\FAcc}{\ensuremath{\mathcal{F}_{\mathsf{acc}}}}
\newcommand{\FComp}{\ensuremath{\mathcal{F}_{\mathsf{comp}}}}
\newcommand{\piComp}{\Pi_{\mathsf{comp}}}
\newcommand{\FAmax}{\ensuremath{\mathcal{F}_{\mathsf{amax}}}}
\newcommand{\FCount}{\ensuremath{\mathcal{F}_{\mathsf{counter}}}}
\newcommand{\Ffilter}{\ensuremath{\mathcal{F}_{\mathsf{filter}}}}
\newcommand{\FZVec}{\ensuremath{\mathcal{F}_{\mathsf{zvec}}}}
\newcommand{\FTrain}{\ensuremath{\mathcal{F}_{\mathsf{Train}}}}
\newcommand{\FpTrain}{\ensuremath{\mathcal{F}_{\mathsf{pTrain}}}}
\definecolor{lightmintbg}{rgb}{.53,.80,.92}

\tikzstyle{maldo} = [rectangle, minimum width=1.0cm, minimum height=0.2cm, text centered, draw=black, fill=orange!75]
\tikzstyle{hdo} = [rectangle,,minimum width=1.0cm, minimum height=0.25cm,text centered, draw=black, fill = lightmintbg!60]
\tikzstyle{malserver} = [rectangle,minimum width=0.6cm, minimum height=0.3cm, text centered, draw=black, fill=red!60]

\tikzstyle{hserver} = [rectangle,minimum width=0.6cm, minimum height=0.3cm, text centered, draw=black, fill=lightmintbg!60]
\tikzstyle{soc} = [rectangle, rounded corners,dashed,minimum width=3.5cm, minimum height=2.6cm, draw=black]

\tikzstyle{nota} = [rectangle, minimum width=7.5cm, minimum height= 0.4cm, font = \small, draw=black]

\tikzstyle{mdnota} = [rectangle, minimum width=0.4cm, minimum height=0.2cm, text centered,font = \small, draw=black, fill=orange!75]

\tikzstyle{msnota} = [rectangle,minimum width=0.4cm, minimum height=0.2cm, text centered, font = \small, draw=black, fill=red!60]

\tikzstyle{hnota} = [rectangle, minimum width=0.4cm, minimum height=0.2cm, text centered,font = \small, draw=black, fill=lightmintbg!60]
\tikzstyle{sarrow} = [ultra thin, <->,latex-latex]
\tikzstyle{darrow} = [thin,->,>=stealth]

\newcommand{\piAddA}{\ensuremath{\Pi_{\add}}}

\newcommand{\piMultA}{\ensuremath{\Pi_{\mult}}}

\newcommand{\mvec}[1]{\ensuremath{{\mathbf{{#1}}}}}

\newcommand{\vb}{\mvec{{b}}}

\newcommand{\vp}{\mvec{{p}}}

\newcommand{\vx}{\mvec{{x}}}
\newcommand{\vy}{\mvec{{y}}}
\newcommand{\vz}{\mvec{{z}}}

\newcommand{\threshold}{\ensuremath{\mathsf{\phi}}}

\newcommand{\piOp}{\ensuremath{\Pi}_{\mathsf{op}}}

\newcommand{\piAMX}{\ensuremath{\Pi_{\mathsf{amx}}}}

\newcommand{\piAcc}{\ensuremath{\Pi_{\mathsf{acc}}}}

\newcommand{\piZVec}{\ensuremath{\Pi_{\mathsf{zvec}}}}

\newcommand{\OB}{\ensuremath{{\mathsf{O^3B}}}}
\newcommand{\cv}[1]{{\ld_{#1}^\text{v}}}
\newcommand{\CV}{{\ld_\text{val}}}
\newcommand{\valacc}{{\text{AccVal}}}
\newcommand{\piEF}{\ensuremath{\Pi_{\mathsf{filter}}}}
\newcommand{\piPred}{\ensuremath{\Pi_{\mathsf{pred}}}}
\newcommand{\pimPred}{\ensuremath{\Pi_{\fm\mathsf{pred}}}}
\newcommand{\safenet}{\text{SafeNet }}

\newlength{\maxlen}

\setlist[description]{style=unboxed,leftmargin=0cm}

\newenvironment{myitemize}{
	\begin{list}{{$\bullet$}}{
			\setlength\partopsep{0pt}
			\setlength\parskip{0pt}
			\setlength\parsep{0pt}
			\setlength\topsep{0pt}
			\setlength\itemsep{2pt}
			\setlength{\itemindent}{0.3pt}
			\setlength{\leftmargin}{10pt}
		}
	}{
	\end{list}
}

\newcounter{itemcount}

\newcommand{\figlab}[1]{\label{fig:#1}}

	\newenvironment{boxfig*}[2]{
		\begin{figure*}[h!]		
			\fontsize{5}{5}\selectfont
			\newcommand{\FigCaption}{#1}
			\newcommand{\FigLabel}{#2}
			\vspace{-.05cm}
			\begin{center}
				\begin{small}			 
					\begin{adjustbox}{max width=\textwidth}
						\begin{tabular}{@{}|@{~~}l@{~~}|@{}}
							\hline
								\rule[-1ex]{0pt}{1ex}\begin{minipage}[b]{.95\linewidth}
									\vspace{1ex}	
								}{%
								\end{minipage}\\
								\hline
							\end{tabular}	
						\end{adjustbox}		
					\end{small}
					\vspace{-0.1cm}
					\caption{\FigCaption}
					\figlab{\FigLabel}
				\end{center}
				\vspace{-.38cm}
			\end{figure*}
		}

			\newenvironment{myboxfig*}[2]{
				\begin{figure*}[!htb]		
					\fontsize{5}{5}\selectfont
					\newcommand{\FigCaption}{#1}
					\newcommand{\FigLabel}{#2}
					\vspace{-.10cm}
					\begin{center}
						\caption{\FigCaption}
						\begin{small}			 
							\begin{adjustbox}{max width=\textwidth}
								\begin{tabular}{@{}|@{~~}l@{~~}|@{}}
									\hline
										\rule[-1ex]{0pt}{1ex}\begin{minipage}[b]{.95\linewidth}
											\vspace{1ex}	
										}{%
										\end{minipage}\\
										\hline
									\end{tabular}	
								\end{adjustbox}		
							\end{small}
							\vspace{-0.25cm}
							\figlab{\FigLabel}
						\end{center}
						\vspace{-.38cm}
					\end{figure*}
				}

				\RequirePackage{mdframed}
				
				\newenvironment{titlebox}[5]
				{\mdfsetup{
						style=#2,
						innertopmargin=1.1\baselineskip,
						skipabove={\dimexpr0.7\baselineskip+\topskip\relax},
						skipbelow={1.5em},needspace=3\baselineskip,
						singleextra={\node[#3,right=10pt,overlay] at (P-|O){~{\sffamily\bfseries #1 }};},%
						firstextra={\node[#3,right=10pt,overlay] at (P-|O) {~{\sffamily\bfseries #1 }};},
						frametitleaboveskip=9em,
						innerrightmargin=5pt
					}
					\newcommand{\TitleCaption}{#4}
					\newcommand{\TitleLabel}{#5}
					\begin{mdframed}[font=\small]
						\setlist[itemize]{leftmargin=13pt}\setlist[enumerate]{leftmargin=13pt}\raggedright%
					}
					{\end{mdframed}
					\vspace{-1.5em}
					{\captionof{figure}{\normalfont \TitleCaption}\label{\TitleLabel}}
					\medskip
				}
				
				\tikzstyle{normal} = [thick, fill=white, text=black, draw, rounded corners, rectangle, minimum height=.7cm, inner sep=3pt]
				\tikzstyle{gray} = [thick, fill=gray!90, text=white, rounded corners, rectangle, minimum height=.7cm, inner sep=3pt]
				
				\mdfdefinestyle{commonbox}{%
					align=center, middlelinewidth=1.1pt,userdefinedwidth=\linewidth
					innerrightmargin=0pt,innerleftmargin=3pt,innertopmargin=5pt,
					splittopskip=7pt,splitbottomskip=5pt
				}
				\mdfdefinestyle{roundbox}{style=commonbox,roundcorner=5pt,userdefinedwidth=\linewidth}
				
				\newenvironment{systembox}[3]
				{ \begin{titlebox}{Functionality \normalfont #1}{roundbox}{normal}{#2}{#3}}
					{\end{titlebox}}
				
				\newenvironment{gsystembox}[3]
				{\vspace{\baselineskip}\begin{titlebox}{Global Functionality \normalfont #1}{roundbox}{normal}{#2}{#3}}
					{\end{titlebox}}
				
				\newenvironment{protocolbox}[3]
				{\begin{titlebox}{Protocol \normalfont #1}{commonbox}{normal}{#2}{#3}}
					{\end{titlebox}}
				
				\newenvironment{algobox}[3]
				{\begin{titlebox}{Algorithm \normalfont #1}{commonbox}{normal}{#2}{#3}}
					{\end{titlebox}}
				
				\newenvironment{reductionbox}[3]
				{\begin{titlebox}{Reduction \normalfont #1}{commonbox}{normal}{#2}{#3}}
					{\end{titlebox}}
				
				\newenvironment{gamebox}[3]
				{\begin{titlebox}{Game \normalfont #1}{commonbox}{gray}{#2}{#3}}
					{\end{titlebox}}
				
				\newenvironment{simulatorbox}[3]
				{\begin{titlebox}{Simulator \normalfont #1}{commonbox}{normal}{#2}{#3}}
					{\end{titlebox}}

				\newenvironment{systembox*}[3]
				{\begin{strip}
						\vspace{\baselineskip}\begin{titlebox}{Functionality \normalfont #1}{roundbox}{normal}{#2}{#3}}
						{\end{titlebox}
				\end{strip}}
				
				\newenvironment{gsystembox*}[3]
				{\begin{strip}
						\vspace{\baselineskip}\begin{titlebox}{Global Functionality \normalfont #1}{roundbox}{normal}{#2}{#3}}
						{\end{titlebox}
				\end{strip}}
				
				\newenvironment{protocolbox*}[3]
				{\begin{strip}
						\begin{titlebox}{Protocol \normalfont #1}{commonbox}{normal}{#2}{#3}}
						{\end{titlebox}
				\end{strip}}
				
				\newenvironment{algobox*}[3]
				{\begin{strip}
						\begin{titlebox}{Algorithm \normalfont #1}{commonbox}{normal}{#2}{#3}}
						{\end{titlebox}
				\end{strip}}
				
				\newenvironment{reductionbox*}[3]
				{\begin{strip}
						\begin{titlebox}{Reduction \normalfont #1}{commonbox}{normal}{#2}{#3}}
						{\end{titlebox}
				\end{strip}}
				
				\newenvironment{gamebox*}[3]
				{\begin{strip}
						\begin{titlebox}{Game \normalfont #1}{commonbox}{gray}{#2}{#3}}
						{\end{titlebox}
				\end{strip}}
				
				\newenvironment{simulatorbox*}[3]
				{\begin{strip}
						\begin{titlebox}{Simulator \normalfont #1}{commonbox}{normal}{#2}{#3}}
						{\end{titlebox}
				\end{strip}}

				\mdfdefinestyle{mycommonbox}{%
					align=center, middlelinewidth=1.1pt,userdefinedwidth=\linewidth
					innerrightmargin=0pt,innerleftmargin=2pt,innertopmargin=3pt,
					splittopskip=7pt,splitbottomskip=3pt
				}
				\mdfdefinestyle{myroundbox}{style=mycommonbox,roundcorner=5pt,userdefinedwidth=\linewidth}
				
				\newenvironment{titlebox*}[5]
				{\mdfsetup{
						style=#2,
						innertopmargin=0.3\baselineskip,
						skipabove={1.2em},
						skipbelow={1em},needspace=3\baselineskip,
						frametitleaboveskip=5em,
						innerrightmargin=5pt
					}
					\newcommand{\TitleCaption}{#4}
					\newcommand{\TitleLabel}{#5}
					\begin{mdframed}[font=\small]
						\setlist[itemize]{leftmargin=13pt}\setlist[enumerate]{leftmargin=13pt}\raggedright%
					}
					{\end{mdframed}
					\vspace{-1.5em}
					{\captionof{figure}{\normalfont \TitleCaption}\label{\TitleLabel}}
					\medskip
				}
				
				\newenvironment{mysystembox*}[3]
				{\begin{strip}
						\vspace{\baselineskip}\begin{titlebox*}{Functionality \normalfont #1}{myroundbox}{normal}{#2}{#3}}
						{\end{titlebox*}
				\end{strip}}
				
				\newenvironment{mygsystembox*}[3]
				{\begin{strip}
						\vspace{\baselineskip}\begin{titlebox*}{Global Functionality \normalfont #1}{myroundbox}{normal}{#2}{#3}}
						{\end{titlebox*}
				\end{strip}}
				
				\newenvironment{myprotocolbox*}[3]
				{\begin{strip}
						\begin{titlebox*}{Protocol \normalfont #1}{mycommonbox}{normal}{#2}{#3}}
						{\end{titlebox*}
				\end{strip}}
				
				\newenvironment{myalgobox*}[3]
				{\begin{strip}
						\begin{titlebox*}{Algorithm \normalfont #1}{mycommonbox}{normal}{#2}{#3}}
						{\end{titlebox*}
				\end{strip}}
				
				\newenvironment{myreductionbox*}[3]
				{\begin{strip}
						\begin{titlebox*}{Reduction \normalfont #1}{mycommonbox}{normal}{#2}{#3}}
						{\end{titlebox*}
				\end{strip}}
				
				\newenvironment{mygamebox*}[3]
				{\begin{strip}
						\begin{titlebox*}{Game \normalfont #1}{mycommonbox}{gray}{#2}{#3}}
						{\end{titlebox*}
				\end{strip}}
				
				\newenvironment{mysimulatorbox*}[3]
				{\begin{strip}
						\begin{titlebox*}{Simulator \normalfont #1}{mycommonbox}{normal}{#2}{#3}}
						{\end{titlebox*}
				\end{strip}}

				\newcommand{\algoHead}[1]{\vspace{0.2em} \underline{\textbf{#1}} \vspace{0.3em}}

				\makeatletter
				\algnewcommand{\ExtendedState}[1]{\State
					\parbox[t]{\dimexpr\linewidth-\ALG@thistlm}{\hangindent=\algorithmicindent\strut\hangafter=3#1\strut}}
				\makeatother
				
				\algnewcommand\algorithmicinput{\textbf{Input:}}
				\algnewcommand\Input{\item[\algorithmicinput]}

				\algrenewcommand{\algorithmiccomment}[1]{{\color{gray}// #1}}
				
				\newcommand{\xmath}[1]{\ensuremath{#1}\xspace}

				\newcommand{\Func}[1][\relax]{\xmath{\mathcal{F}_{\textsc{#1}}}}

\newcommand{\commentH}[1] {\textcolor{blue} {{\sf (}{\sl{#1}} {\sf - Harsh)}}}

\begin{document}


\author{\IEEEauthorblockN{Harsh Chaudhari\IEEEauthorrefmark{1},
Matthew Jagielski\IEEEauthorrefmark{2},
 Alina Oprea\IEEEauthorrefmark{1}}
\IEEEauthorblockA{\IEEEauthorrefmark{1}Northeastern University,
\IEEEauthorrefmark{2}Google Research}}


\IEEEoverridecommandlockouts
\makeatletter\def\@IEEEpubidpullup{6.5\baselineskip}\makeatother

\title{SafeNet: The Unreasonable Effectiveness of Ensembles in Private Collaborative Learning}

\maketitle

\begin{abstract}
Secure multiparty computation (MPC) has been proposed to allow multiple mutually distrustful data owners to jointly train machine learning (ML) models on their combined data. However, by design, MPC protocols faithfully compute the training functionality, which the adversarial ML community has shown to leak private information and can be tampered with in poisoning attacks. In this work, we argue that model ensembles, implemented in our framework called SafeNet, are a highly MPC-amenable way to avoid many adversarial ML attacks. The natural partitioning of data amongst owners in MPC training allows this approach to be highly scalable at training time, provide provable protection from poisoning attacks, and provably defense against a number of privacy attacks. We demonstrate SafeNet's efficiency, accuracy, and resilience to poisoning on several machine learning datasets and models trained in end-to-end and transfer learning scenarios. For instance, SafeNet reduces  backdoor  attack success significantly, while achieving $39\times$ faster training and $36 \times$ less communication than the four-party MPC framework of Dalskov et al.~\cite{DEK21}. Our experiments show that ensembling retains these benefits even in many non-iid settings. The simplicity, cheap setup, and robustness properties of ensembling make it a strong first choice for training ML models privately in MPC.

\end{abstract}






\section{Introduction}

Machine learning (ML) has been successful in a broad range of application areas such as  medicine, finance, and recommendation systems. Consequently, technology companies such as Amazon, Google, Microsoft, and IBM  provide machine learning  as a service (MLaaS) for ML training and prediction. In these services, data owners  outsource their ML computations to a set of more computationally powerful servers. However, in many instances, the client data used for ML training or classification is sensitive and may be subject to privacy requirements. Regulations such as GDPR, HIPAA and PCR, data sovereignty issues, and user privacy concern are common reasons preventing organizations from collecting user data and training more accurate ML models. These privacy requirements have led to the design of privacy-preserving ML training methods, including the use of secure multiparty computation (MPC). 




Recent literature in the area of MPC for ML proposes privacy-preserving machine learning (PPML) frameworks~\cite{MohasselZ17, MR18, WaghGC18, DEK20, Delphi20, WTBKMR21, DEK21, AEV21, Cerebro21,Piranha22} for training and inference of various machine learning models such as logistic regression, neural networks, and random forests. In these models, data owners outsource shares of their data to a set of servers and the servers run MPC protocols for ML training and prediction. An implicit assumption for security is that the underlying datasets provided by data owners during  training have not been influenced by an adversary. However, research in adversarial machine learning has shown that data poisoning attacks pose a high risk to the integrity of trained ML models~\cite{BNL12,JOBLR18,GLDG19,geiping2020witches}. Data poisoning becomes a particularly relevant threat in PPML systems, as multiple data owners contribute secret shares of their datasets for jointly training a ML model inside the MPC, and  poisoned samples cannot be easily detected. Furthermore, the guarantees of MPC provide privacy against an adversary observing the communication in the protocol, but does not protect against any sensitive information leaked by the model about its training set. Many privacy attacks are known to allow inference on machine learning models' training sets, and protecting against these attacks is an active area of research.





In this paper, we study the impact of these adversarial machine learning threats on standard MPC frameworks for private ML training. Our first observation is that the security definition of MPC for private ML training does not account for data owners with poisoned data. Therefore, we extend the security definition by considering an adversary who can poison the datasets of a subset of  owners, while at the same time controlling a subset of the servers in the MPC protocol. 
Under our threat model, we empirically demonstrate that poisoning attacks are a significant threat to the setting of private ML training. 
We show the impact of backdoor~\cite{GLDG19,CLLLS17} and targeted~\cite{koh2017understanding,geiping2020witches} poisoning attacks on four MPC frameworks and five datasets, using logistic regression and neural networks models. We show that with control of just a single owner and its dataset (out of a set of 20 owners contributing data for training), the adversary achieves $100\%$ success rate  for a backdoor attack, and higher than $83\%$ success rate for a targeted attack. These attacks are stealthy and cannot be detected by simply monitoring standard ML accuracy metrics.


To mitigate these attacks, we apply ensembling technique from ML, implemented in our framework called SafeNet, which, in the collaborative learning setting we consider, is an effective defense against poisoning attacks, while also simultaneously preventing various types of privacy attacks. Rather than attempting to implement an existing poisoning defense in MPC, we observe that the structure of the MPC threat model permits a more general and efficient solution. Our main insight is to require individual data owners to train ML models locally, based on their own datasets, and secret share the resulting ensemble of models in the MPC. We filter out local models with low accuracy on a validation dataset, and use the remaining models to make predictions using a majority voting protocol performed inside the MPC. While this permits stronger model poisoning attacks, the natural partitioning of the MPC setting prevents an adversary from poisoning more than a fixed subset of the models, resulting in a limited number of poisoned models in the ensemble. We perform a detailed analysis of the robustness properties of SafeNet, and provide lower bounds on the ensemble's accuracy based on the error rate on the local models in the ensemble and the number of poisoned models, as well as a prediction certification procedure for arbitrary inputs. The bounded contribution of each local model also gives a provable privacy guarantee for SafeNet. Furthermore, we show empirically that SafeNet successfully mitigates backdoor and targeted poisoning attacks, while retaining high accuracy on the ML prediction tasks. In addition, our approach is efficient, as ML model training is performed locally by each data owner, and only the ensemble filtering and prediction protocols are performed in the MPC. This provides large performance improvements in ML training compared to existing PPML frameworks, while simultaneously mitigating poisoning attacks. For instance, for one neural network model, SafeNet performs training $39\times$ faster than the \cite{DEK21} PPML  protocol and requires $36 \times$ less communication. Finally, we investigate settings with diverse data distributions among owners, and evaluate the accuracy and robustness of SafeNet under multiple data imbalance conditions.

To summarize, our contributions are as follows:

\smallskip
\myparagraph{Adversarial ML-aware Threat Model for Private Machine Learning} We extend the MPC security definition for private machine learning to encompass the threat of data poisoning attacks and privacy attacks. In our threat model, the adversary can poisoned a subset $t$ out of $m$  data owners, and control  $T$ out of $N$ servers participating in the MPC. The attacker might also seek to learn sensitive information about the local datasets through the trained model. 

\smallskip
\myparagraph{SafeNet Ensemble Design} We propose SafeNet, which adapts  ensembling technique from ML to the collaborative MPC setting by having data owners train models locally and aggregation of predictions is performed securely inside the MPC. We show that this procedure gives provable privacy and security guarantees, which improves as models become more accurate. We also propose various novel extensions to this ensembling strategy which make SafeNet applicable to a wider range of training settings (including transfer learning and accommodating computationally restricted owners). SafeNet's design is agnostic to the underlying MPC framework and we show it can be instantiated over four different MPC frameworks, supporting two, three and four servers.


\smallskip
\myparagraph{Comprehensive Evaluation} We show the impact of existing backdoor and targeted poisoning attacks on several existing PPML systems~\cite{DSZ15,AFLNO16,DEK21} and five datasets, using logistic regression and neural network models. We also empirically demonstrate the resilience of SafeNet against these attacks, for an adversary compromising up to 9 out of 20 data owners. We report the gains in training time and communication cost for SafeNet compared to existing PPML frameworks. Finally, we compare SafeNet with state-of-the-art defenses against poisoning in federated learning~\cite{CJG21} and show its enhanced certified robustness even under non-iid data distributions. 






\section{Background and Related Work}

We provide background on secure multi-party computation and poisoning attacks in ML, and discuss related work in the area of adversarial ML and MPC. 

\subsection{Secure Multi-Party Computation} 

Secure Multi-Party Computation (MPC)~\cite{Yao82,BGW88,GMW87,IKNP03,DPSZ12}  allows a set of $n$ mutually distrusting parties to compute a joint function $f$, so that collusion of any $t$ parties cannot modify the  output of computation (\emph{correctness}) or learn any information beyond what is revealed by the output (\emph{privacy}). The area of  MPC can be  categorized into honest majority \cite{BGW88,MRZ15,AFLNO16,CCPS19, BCPS20} and dishonest majority \cite{Yao82,DPSZ12,SPDZ2,MohasselF06a,GMW87}. The settings of 
two-party computation (2PC) \cite{Yao82,LindellP07,Lindell16,NielsenO16}, three parties (3PC) \cite{ABFLLNOWW17,AFLNO16,MRZ15}, and four parties (4PC) \cite{IshaiKKP15, GordonR018,CRS20, DEK21} have been widely studied as they provide efficient protocols. 
Additionally, recent works in the area of privacy preserving ML propose training and prediction frameworks~\cite{MohasselZ17, MR18, WaghGC18, KRCGRS20, RRKCGRS20, WTBKMR21, AEV21, PSSY21} built on top of the above MPC settings. Particularly, most of the frameworks are deployed in the outsourced computation setting where the data is secret-shared to a set of servers which perform training and prediction using MPC.     


\subsection{Data Poisoning Attacks} \label{sec:DP}

In a data poisoning attack, an adversary controls a subset of the training dataset, and uses this to influence the model trained on that training set. 
In a backdoor attack~\cite{NKS06, GLDG19, CLLLS17}, an adversary seeks to add a ``trigger'' or backdoor pattern into the model. The trigger is a perturbation in feature space, which is applied to poisoned samples in training to induce misclassification on backdoored samples at testing. In a targeted attack~\cite{koh2017understanding, koh2018stronger, suciu2018does}, the adversary's goal is to change the classifier prediction for a small number of specific test samples. Backdoor  and targeted attacks can be difficult to detect, due to the subtle impact they have on the ML model.



\subsection{Related Work}
While both MPC and adversarial machine learning have been the topic of fervent research, work connecting them is still nascent. We are only aware of several recent research papers that attempt to bridge these areas. Recent works \cite{lehmkuhl2021muse, CGOS22} show that MPC algorithms applied at test time can be compromised by malicious users, allowing for  efficient \emph{model extraction} attacks. Second, Escudero et al.~\cite{escudero2021adversarial} show that running a semi-honest MPC protocol with malicious parties can result in backdoor attacks in the resulting SVM model. Both these works, as well as our own, demonstrate the difficulty of aligning the guarantees of MPC with the additional desiderata of adversarial machine learning. We demonstrate the effectiveness of data poisoning attacks in MPC for neural networks and logistic regression models, and propose a novel ensemble training algorithm in SafeNet to defend against poisoning attacks in MPC.

Model ensembles have been proposed as a defense for ML poisoning and privacy attacks in prior work in both the centralized training setting~\cite{biggio2011bagging, jia2020intrinsic} and the collaborative learning setting. Compared to centralized approaches, which process a single dataset, we are able to leverage the trust model of MPC, which limits the number of poisoned models in the ensemble and can provide stronger robustness and privacy guarantees. Ensembles have also been proposed in MPC to protect data privacy~\cite{choquettechoo2021capc} and in federated learning to provide poisoning robustness~\cite{CJG21}. Our work provides a stronger privacy analysis, protecting from a broader range of threats than \cite{choquettechoo2021capc}, and additionally offers robustness guarantees. We provide a more detailed comparison with these approaches in Section~\ref{sec:otherensembles}.


\section{SafeNet: Using Ensembles in MPC}
We describe here our threat model and show how to implement ensembles in  MPC. We then show that ensembling gives us provable robustness to poisoning and privacy adversaries.

\subsection{Threat Model}

{
	\begin{figure}[t ]
		\vspace{-1.5mm}	
		\begin{tikzpicture}[node distance= 0.7cm]
			\node(maldo1)[maldo]{ \footnotesize $\usr_1$ };
			\node(maldoi-1)[maldo, below of=maldo1]{\footnotesize $\usr_t$ };
			\node(hdoi)[hdo, below of=maldoi-1]{\footnotesize $\usr_{t+1}$ };
			\node(hdom)[hdo, below of=hdoi]{\footnotesize $\usr_m$ };
			
			\node(mals1)[malserver, right of = maldo1, xshift=4.6cm, yshift = -0.3cm]{\footnotesize S$_1$};
			
			\node(mals2)[malserver, right of = mals1, yshift = -0.8cm]{\footnotesize S$_2$};
			
			\node(malsT)[malserver, below  of = mals1, yshift = -1.0cm]{ \footnotesize S$_{\scaleto{T}{3pt}}$};
			
			\node(hsT1)[hserver, left  of = malsT, xshift = -0.4cm]{S$_{\footnotesize \scaleto{T+1}{3pt}}$};
			
			\node(hsN)[hserver, xshift = -0.4cm , left of = mals1]{S$_{\footnotesize \scaleto{N}{3pt}}$};
			
			\node(hsN1)[hserver, yshift = -0.8cm , left of = hsN]{S$_{\footnotesize \scaleto{N-1}{3pt}}$};
			
			\node(soc)[soc, left of = hsN, label= {\footnotesize SOC Paradigm}, xshift = 1.25cm, yshift = -0.85cm]{};
			
			\node(nota)[nota, below of = soc, xshift = - 2.1cm, yshift = -1.1cm]{};
			
			\node(mdnota)[mdnota, label = right: {\footnotesize Poisoned}, left of = nota, xshift = -2.3cm]{};
			\node(msnota)[msnota, right of = mdnota, label = right: {\footnotesize Corrupted}, xshift = 1.78cm]{};
			\node(hnota)[hnota, right of = msnota,label = right: {\footnotesize Honest}, xshift = 1.8cm]{};
			
			\path (maldo1) -- node[auto=false]{\ldots} (maldoi-1);
			\path (hdoi) -- node[auto=false]{\ldots} (hdom);
			\path (hdoi) -- node[auto=false]{\ldots} (hdom);
			
			\path[draw, densely dotted,  bend left] (mals2) edge (malsT);
			\path[draw, densely dotted,  bend left] (hsT1) edge (hsN1);
			
			\draw[sarrow](mals1) -- (mals2);
			\draw[sarrow](mals1) -- (malsT);
			\draw[sarrow](mals1) -- (hsT1);
			\draw[sarrow](mals1) -- (hsN1);
			\draw[sarrow](mals1) -- (hsN);
			
			\draw[sarrow](mals2) -- (malsT);
			\draw[sarrow](mals2) -- (hsT1);
			\draw[sarrow](mals2) -- (hsN1);
			\draw[sarrow](mals2) -- (hsN);
			
			\draw[sarrow](malsT) -- (hsT1);
			\draw[sarrow](malsT) -- (hsN1);
			\draw[sarrow](malsT) -- (hsN);
			
			\draw[sarrow](hsT1) -- (hsN1);
			\draw[sarrow](hsT1) -- (hsN);
			
			\draw[sarrow](hsN1) -- (hsN);
			
			\draw[darrow](maldo1) -- (soc);
			\draw[darrow](maldoi-1) -- (soc);
			\draw[darrow](hdoi) -- (soc);
			\draw[darrow](hdom) -- (soc);
		\end{tikzpicture}
		\caption{Threat model considered in our setting. The adversary $\Apsc$ can poison at most $t$ out of $m$ data owners and corrupt at most $T$ out of $N$ servers participating in the MPC computation. $\usr_i$ and $\srvr_j$ denote the $i^{th}$ data owner and $j^{th}$ server.}
		\label{fig:AS}
		\vspace{-1.5mm}
\end{figure}
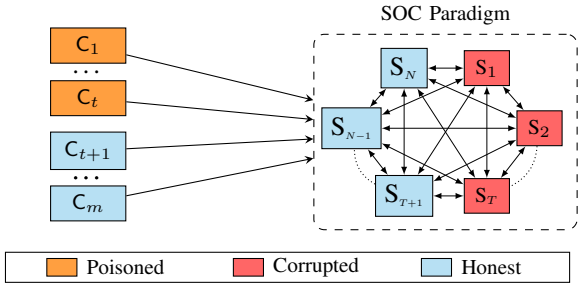}

\myparagraph{Setup} We consider a set of $m$ data owners $C= \cup_{k=1}^{m} \usr_k$ who wish to train a joint machine learning model $\fm$ on their combined dataset $\fd =  \cup_{k=1}^{m} \ld_k$. We adopt the Secure Outsourced Computation (SOC) paradigm~\cite{MohasselZ17, MR18, WaghGC18, BCPS20, RRKCGRS20, WTBKMR21, AEV21, DEK20, DEK21} for training model $\fm$ privately, where the owners secret-share their respective datasets to a set of  outsourced servers, who execute the MPC protocols to train $\fm$. The final output   is a trained model in secret-shared format among the servers.  
A single training/testing sample is expressed as $(\ssample_i, \slbl_i)$, where $\ssample_i$ is the input feature vector and $\slbl_i$ is its corresponding true label or class. We use $\ld_k = (\dataset_{\sss{k}}, \lbl_{\sss{k}})$ to denote dataset of data owner $\usr_k$ participating in the training process. Matrix $\dataset_{\sss{k}}$ denotes a feature matrix where the number of rows represent the total training samples possessed by $\usr_k$ and $\lbl_{\sss{k}}$ denotes the corresponding vector of true labels.

\smallskip
\myparagraph{Adversary in the SOC} Given  a set $S = \{\srvr_1,\ldots,\srvr_N\}$ of servers, we define an adversary $\Asoc$,  similar to prior work~\cite{MohasselZ17, MR18, RRKCGRS20, WTBKMR21, AEV21, DEK21}.  $\Asoc$ can statically corrupt a  subset $S_T \subset S$ of servers of size at most $T<N$. The exact values of $N$ and $T$ are dependent on the  MPC protocols  used for training the ML model privately. We experiment with two-party, three-party, and four-party protocols with one corrupt server. MPC defines two main adversaries: i) \emph{Semi-honest}: Adversary follows a given protocol, but tries to derive additional information from the messages received from other parties during the protocol; ii) \emph{Malicious}: Adversary has the ability to arbitrarily deviate during the execution of the protocol. 

\smallskip
\myparagraph{Security Definition} MPC security is defined using the real world - ideal world paradigm~\cite{C00}.
In the real world, parties participating in the MPC interact during the execution of a protocol $\pi$ in presence of an adversary $\Adv$. Let $\mathsf{REAL}[\mathbb{Z}, \Adv, \pi, \lambda]$ denote the output of the environment $\mathbb{Z}$ when interacting with $\Adv$ and the honest parties, who execute $\pi$ on security parameter $\lambda$. Effectively, $\mathsf{REAL}$ is a function of the inputs/outputs and messages sent/received during the protocol. In the ideal world, the parties simply forward their inputs to a trusted functionality $\mathcal{F}$ and forward the functionality's response to the environment.
Let $\mathsf{IDEAL}[\mathbb{Z}, \Sim, \mathcal{F}, \lambda]$ denote the output of the environment $\mathbb{Z}$ when interacting with adversary $\Sim$ and honest parties who run the  protocol in presence of $\mathcal{F}$ with security parameter $\lambda$. The security definition states that the views of the adversary in the real and ideal world are indistinguishable: 

\begin{definition} \label{def:mpc_sec}
	A  protocol $\pi$ securely realizes  functionality $\mathcal{F}$   if for all environments $\mathbb{Z}$ and any adversary of type $\Asoc$, which  corrupts a subset $S_T$ of servers of size at most $T < N$ in the real world, then there exists a simulator $\Sim$ attacking the ideal world, such that $\mathsf{IDEAL}[\mathbb{Z}, \Sim, \mathcal{F}, \lambda] \approx \mathsf{REAL}[\mathbb{Z}, \Asoc, \pi, \lambda]$.
\end{definition}

\smallskip
\myparagraph{Poisoning Adversary} Existing threat models for training ML models privately assume that the local datasets contributed towards training  are not under the control of the adversary. However, data poisoning attacks have been shown to be a real threat when ML models are  trained on crowdsourced data or data coming from untrusted sources~\cite{BNL12,GBDPLR17,JOBLR18}. Data poisoning becomes a particularly relevant risk in PPML systems, in which data owners contribute their own datasets for training a joint ML model. Additionally, the datasets are secret shared among the servers participating in the MPC, and potential poisoned samples (such as backdoored data) cannot be easily detected by the servers running the MPC protocol.

To account for such attacks, we define a poisoning adversary $\Ap$  that can  poison  a subset of local datasets of size at most $t <m$. Data owners with poisoned data are called {\em poisoned owners}, and we assume that the adversary can coordinate with the poisoned owners to achieve a certain adversarial goal. For example, the adversary can mount a backdoor attack, by selecting a backdoor pattern and poison the datasets under its control with the particular backdoor pattern.

{\em Poisoning Robustness:} We consider an ML model to be robust against a poisoning adversary $\Ap$, who poisons the datasets of $t$ out of $m$ owners, if it generates correct class predictions on new samples with high probability. We provide bounds on the level of poisoning tolerated by our designed framework to ensure robustness.


\smallskip
\myparagraph{Our Adversary} We now define a new adversary $\Apsc$  for our threat model (Figure~\ref{fig:AS}) that corrupts servers in the MPC and poisons the owners' datasets:  
\begin{myitemize}
	\smallskip
	\item[--] $\Apsc$ plays the role of $\Ap$ and poisons $t$ out of $m$ data owners that secret share their training data to the  servers.
	\item[--] $\Apsc$ plays the role of $\Asoc$ and corrupts $T$ out $N$ servers taking part in the MPC computation.

\end{myitemize}

\smallskip
Note that the poisoned owners that $\Apsc$ controls do not interfere in the execution of the MPC protocols after secret-sharing their data and also do not influence the honest owners. 

\smallskip
\myparagraph{Functionality  $\FpTrain$} Based on our newly introduced threat  model, we construct a new functionality $\FpTrain$ in Figure 2 to accommodate poisoned data. 

\begin{systembox}{$\FpTrain$}{Ideal Functionality for ML training with data poisoning}{func:pml}   
	\label{fig:ideal_func}
	
	\algoHead{Input:} $\FpTrain$ receives secret-shares of $\ld_i$ and $a_i$ from each owner $\usr_i$, where $\ld_i$ is a dataset and  $a_i$  an auxiliary input.
	
	\algoHead{Computation: } On receiving inputs from the owners, $\FpTrain$ computes $ O = f(\ld_1, . . . , \ld_m,a_1, \ldots,a_m)$, where $f$ and $O$ denotes the training algorithm and the output of the algorithm respectively.
	
	\algoHead{Output:} $\FpTrain$ constructs secret-shares of $O$ and sends the appropriate shares to the servers.

\end{systembox}

\smallskip
\myparagraph{Security against $\Apsc$} A training protocol $\pistr$ is secure against adversary  $\Apsc$  if: (1)  $\pistr$ securely realizes functionality $\FpTrain$ based on Definition~\ref{def:mpc_sec}; and (2)  the   model trained inside the MPC provides poisoning robustness against data poisoning attacks.

Intuitively, the security definition ensures that $\Apsc$ learns no information about the honest owners' inputs when $T$ out of $N$ servers are controlled by the adversary, while the  trained model provides poisoning robustness  against a subset of $t$ out of $m$ poisoned owners.

\smallskip
\myparagraph{Extension to Privacy Adversary}
While MPC guarantees no privacy leakage during the execution of the protocol, it makes no promises about privacy leakage that arises by observing the output of the protocol. This has motivated a combination of differential privacy guarantees with MPC algorithms, to protect against privacy leakage for both the intermediate execution as well as the output of the protocol. For this reason, we also consider adversaries seeking to learn information about data owners' local datasets by observing the output of the model, as done in membership inference \cite{shokri2017membership, yeom2018privacy, LiRA} and property inference attacks \cite{Ganju18, Zhang21, DistributionInference}. Recent works have used data poisoning as a tool to further increase privacy leakage~\cite{TruthSerum,Chase21,SNAP22} of the trained models.
Consequently, we can extend our threat model to accommodate a stronger version of $\Apsc$ that is also capable of performing privacy attacks by observing the output of the trained model.


\subsection{\safenet Overview}\label{sec:SNoverview}
\begin{figure*}[t] 
	\includegraphics[scale=0.3]{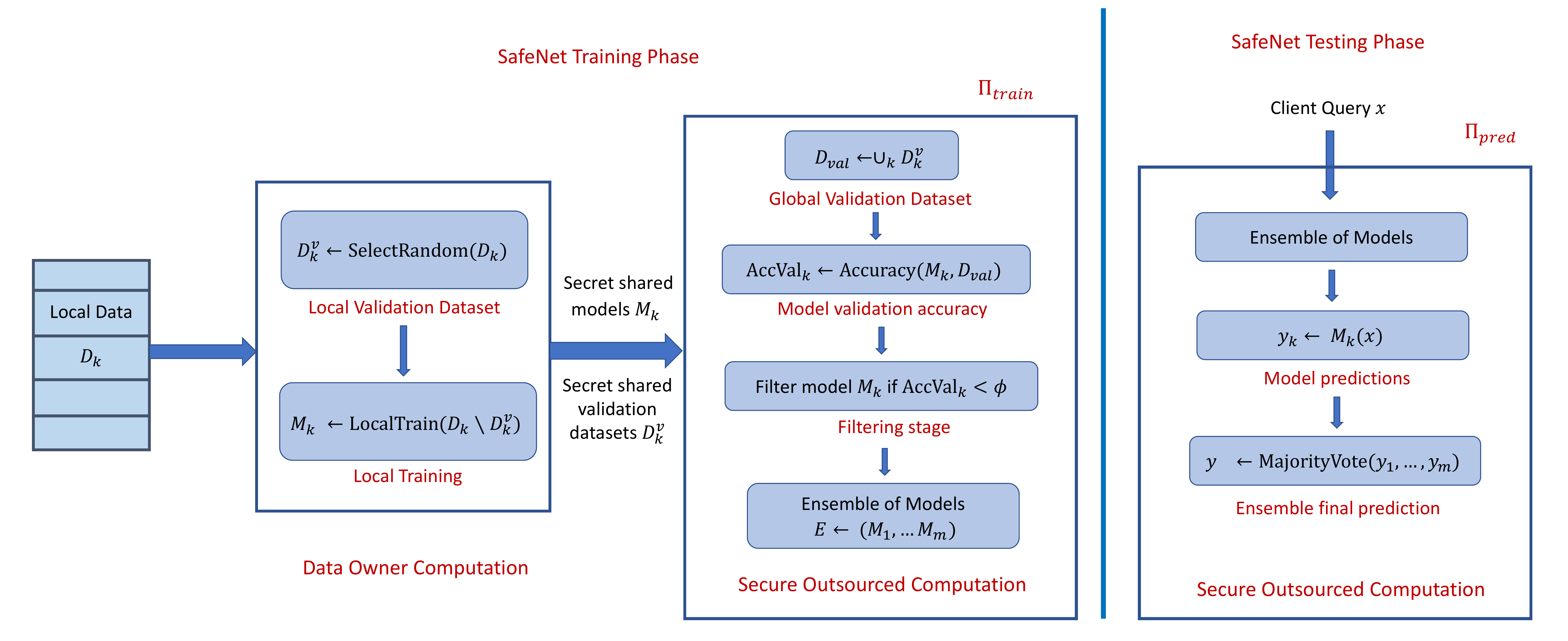}
	\caption{Overview of the Training and Inference phases of the \safenet Framework.}
	\label{fig:EF}
\end{figure*}

Given our threat model in Figure~\ref{fig:AS}, existing PPML  frameworks provide security against an $\Asoc$ adversary, but they are not designed to handle an  $\Apsc$ adversary. We show experimentally in Section~\ref{sec:exp} that PPML frameworks for private training are susceptible to data poisoning attacks.  While it would be possible to remedy this by implementing specific poisoning defenses (see Section~\ref{sec: otherDefenses} for a discussion of these approaches), we instead show that it is possible to take advantage of the bounded poisoning capability of $\Apsc$ to design a more general and efficient defense. Intuitively, existing approaches train a single model on all local datasets combined, causing the model's training set to have a large fraction of poisoned data ($t/m$), which is difficult to defend against. Instead, we design SafeNet, a new protocol which uses ensemble models to realize our threat model and provide security against $\Apsc$. In addition to successfully mitigating data poisoning attacks, SafeNet provides more efficient training than existing PPML and comparable prediction accuracy.


Figure~\ref{fig:EF} provides an overview of the training and inference phases of SafeNet. \safenet trains an ensemble $E$ of multiple models in protocol $\pistr$, where each model $\fm_k \in E$ is trained locally by the data owner $\usr_k$ on their dataset. This partitioning prevents poisoned data from contributing to more than $t$ local models. Each data owner samples a local validation dataset and trains the local model $\fm_k$ on the remaining data. The local models and validation datasets are secret shared to the outsourced servers. We note that this permits arbitrarily corrupted models, and poisoned validation datasets, but SafeNet's structure still allows it to tolerate these corruptions. In the protocol running inside the MPC, the servers jointly implement a filtering stage for identifying models with low accuracy on the combined validation data  (below a threshold $\phi$) and excluding them from the ensemble. The output of training is a secret share of each model in the trained ensemble $E$. 



In the inference phase, SafeNet implements protocol $\piPred$, to compute the prediction $y_k$ of each shared model $\fm_k$ on test input $x$ inside the MPC. The servers  jointly perform majority voting to determine the most common predicted class $y$ on input $x$, using only the models which pass the filtering stage. An optional feature of SafeNet is to add noise to the majority vote to enable user-level differential privacy protection, in addition to poisoning robustness. 



Our \safenet protocol leverages our threat model, which assumes that only a set of at most $t$ out of $m$ data owners are poisoned. This ensures that an adversary only influences a limited set of models in the ensemble, while existing training protocols would train a single poisoned global model. We provide bounds for the exact number of poisoned owners $t$ supported by our ensemble in Theorem~\ref{thm:LB}. Interestingly, the bound depends on the number of data owners $m$, and the maximum error made by a clean model in the ensemble.  The same theorem also lower bounds the probability that the ensemble predicts correctly under data poisoning performed by the $t$ poisoned owners, and we validate experimentally that, indeed, \safenet provides resilience to stealthy data poisoning attacks, such as backdoor and targeted attacks. Another advantage of \safenet is that the training time to execute the MPC protocols in the SOC setting is drastically reduced as each $\fm_k \in E$ can be trained locally by the respective owner. 
We detail below the algorithms for training and inference in SafeNet.

\subsection{\safenet Training and Inference} \label{sec:train}

To train the ensemble in SafeNet, we present our proposed ensemble method in Algorithm 1. We  discuss the realization in MPC  later in Appendix~\ref{sec:mpc_inst}. Each owner $\usr_k$ separates out a subset of its training dataset $\cv{k} \in \ld_k$ and then trains its model $\fm_k$ on the remaining dataset $\ld_k \setminus \cv{k}$. The trained model $\fm_k$ and validation dataset $\cv{k}$ is then secret-shared to the servers. The combined validation dataset is denoted as $\CV = \bigcup\limits_{i=1}^{m} \cv{i}$. We assume that all users contribute equal-size validation sets to $\CV$. During the filtering stage inside the MPC, the validation accuracy $\valacc$ of each model is jointly computed on $\CV$. If the resulting accuracy for a model is below threshold $\threshold$, the model is excluded from the  ensemble.

The  filtering step is used to separate the models with low accuracy, either contributed by a poisoned owner, or by an owner holding non-representative data for the prediction task.
  Under the assumption that the majority of owners are honest, it follows that the majority of validation samples are correct. If  $\usr_k$ is honest, then the corresponding $\fm_k$ should have a high validation accuracy on $\CV$, as the corresponding predicted outputs would most likely agree with the samples in $\CV$. In contrast, the predictions by a poisoned model $\fm_k$ will likely not match the samples in $\CV$.  In Appendix \ref{sec:SafeNetAnalysisProof}, we compute a lower bound on the size of the validation dataset as a function of  the number of poisoned owners $t$ and filtering threshold $\threshold$, such that all clean models pass the filtering stage with high probability even when a subset of the cross-validation dataset $\CV$ is poisoned. 

Given protocol $\pistr$ that securely realizes Algorithm 1 inside the MPC (described  in Appendix~\ref{sec:mpc_inst}), we argue security as follows:

\begin{restatable}{theorem}{SPtrain} \label{thm:SNtrain}
Protocol $\pistr$ is secure against  adversary $\Apsc$ who poisons $t$ out of $m$ data owners and  corrupts $T$ out of $N$ servers.
\end{restatable}

\noindent
The proof of the theorem will be given in Appendix~\ref{app:secproof} after we introduce the details of MPC instantiation and how  protocol $\pistr$ securely realizes $\FpTrain$ in Appendix \ref{sec:sectrain}. 

During inference, the prediction of each model $\fm_k$ is generated and the servers aggregate the results to perform majority voting. Optionally, differentially private noise is added to the sum to offer user-level privacy guarantees.  The secure inference protocol \piPred\ in MPC and its proof of security is given in Appendix~\ref{sec:mpc_inst} and \ref{app:secproof} respectively.

\begin{algorithm}[t]
	
	\caption*{{\bf Algorithm 1} SafeNet Training Algorithm} \label{alg:ensemblefiltering}
	\begin{algorithmic}
		
			\State Input: $m$ data owners, each owner $\usr_k$'s   dataset $\ld_k$. 	
			 
			 \smallskip
			 \noindent
			 \Comment{\footnotesize Owner's local computation in plaintext format} 
			\State -- For $k \in [1,m]:$	
			\begin{itemize}
				\item[-] Separate out $\cv{k}$ from $\ld_k$. Train $\fm_k$ on $\ld_k \setminus \cv{k}$.
				\item[-] Secret-share $\cv{k}$ and $\fm_k$ to servers.
			\end{itemize}
		
		
		    \smallskip
		    \noindent
			\Comment{\footnotesize MPC computation in secret-shared format}
			\State -- Construct a common validation dataset $\CV = \cup_{i=1}^{m} \cv{i}$.
			
			\State -- Construct ensemble of models $ E = \lbrace \fm_i \rbrace_{i=1}^{m}$
		
			\State -- Initialize a vector $\valvec$ of zeros and of size $m$. 
		
			\State  -- For $k \in [1,m]:$ \Comment{\footnotesize Ensemble Filtering}
			\begin{itemize}
				\item[-] $\valacc_k = Accuracy(\fm_k, \CV)$ 
			
			
			
				\item[-] If $\valacc_k> \threshold$: Set $\valvec_k = 1$  
				
			
			\end{itemize}
			\State \Return $E$ and $\valvec$
		\end{algorithmic}
\label{alg:train}
\end{algorithm}

\subsection{SafeNet Analysis} \label{sec:Analysis}
Here, we demonstrate the accuracy, poisoning robustness and privacy guarantees that SafeNet provides.  
We first show how to lower bound SafeNet's test accuracy given that each clean model in the ensemble reaches a certain accuracy level. We also give certified robustness and user-level privacy guarantees. All of our guarantees improve as the individual models become more accurate, making the ensemble agree on correct predictions more frequently.


\smallskip
\myparagraph{Robust Accuracy Analysis} 
We provide lower bounds on SafeNet accuracy, assuming that at most $t$  out  $m$ models in  the SafeNet ensemble $E$ are poisoned, and the clean models have independent errors, with maximum error rate $p < 1-\threshold$, where $\threshold$ is the filtering threshold. 


\smallskip
\noindent
{\bf Theorem. (Informal)} {\em  Let  $\Apsc$ be an adversary who poisons at most $t$ out of $m$ data owners and corrupts $T$ out of $N$ servers. Assume that the filtered ensemble $E$ has at least $m-t$ clean models, each with a maximum error rate of $p < 1-\threshold$. If the number of  poisoned owners is at most $\frac{m(1-2p)}{2(1-p)}$, ensemble $E$ correctly classifies new samples with high probability, which is a function of  $m$, $\phi$, $t$ and $p$.}

The formal theorem and the corresponding proof can be found in Appendix \ref{sec:SafeNetAnalysisProof}. 


\smallskip
\myparagraph{Poisoning Robustness Analysis}
Our previous theorem demonstrated that SafeNet's accuracy on in-distribution data is not compromised by poisoning. Now, we show that we can also certify robustness to poisoning on a per-sample basis for arbitrary points, inspired by certified robustness techniques for adversarial example robustness~\cite{cohen2019certified}. In particular, Algorithm 2 describes a method for certified prediction against poisoning, returning the most common class $y$ predicted by the ensemble on a test point $x$, as well as a bound on the number of poisoning owners $t$ which would be required to modify the predicted class.
%
\begin{algorithm}
	\begin{algorithmic}
		\smallskip
		\State Input: $m$ data owners; Ensemble of models $E = \lbrace \fm_i \rbrace_{i=1}^{m}$; Testing point $x$; Differential Privacy parameters $\varepsilon, \delta$.
		\State $\textsc{Counts} = \sum_{i=1}^m \fm_i(x) { \color{purple} ~+ ~\textsc{DPNoise}(\varepsilon, \delta)}$  
		\State $y, c_y = \textsc{MostCommon}(\textsc{Counts})$ 
		\Comment{most common predicted class with noisy count}
		\State $y', c_{y'} = \textsc{SecondMostCommon}(\textsc{Counts})$ \Comment{second most common predicted class with count}
		\State $t = \lceil(c_y-c_{y'})/2\rceil - 1$
		\State \Return $y, t$
	\end{algorithmic}
	\caption*{{\bf Algorithm 2} Certified Private Prediction $\textsc{PredGap}~(E,x)$} 
	\label{alg:predgap}
\end{algorithm}

We first analyze the poisoning robustness when privacy of aggregation is not enabled in the following theorem. 
\begin{restatable}{theorem}{RP} \label{thm:robustpredcorrect} Let $E$ be an ensemble of models trained on  datasets $D=\{D_1,\dots,D_m\}$. Assume that on an input $x$, the ensemble generates prediction $y=E(x)$ without $\textsc{DPNoise}$ and Algorithm 2 outputs $(y, t)$. Moreover, assuming an adversary $\Apsc$ who poisons at most $t$ data  owners, the resulting $E'$ trained on poisoned data $D'$ generates the same prediction on $x$ as $E$: $E'(x)=y$.
\end{restatable}

\begin{proof}
	If an adversary's goal were to cause $y'$ to be predicted on input $x$, their most efficient strategy is to flip $y$ predictions to $y'$. If $y$ were the ensemble prediction, it must have at least $\lfloor\frac{c_y+c_{y'}}{2}\rfloor$ model predictions, and the second most common prediction $y'$ would have at most $\lfloor\frac{c_y+c_{y'}}{2}\rfloor$ model predictions. Corrupting these predictions then requires flipping at least $(c_y-c_{y'})/2$ predictions from $y$ to $y'$. Overall, this requires at least $\lceil(c_y-c_{y'})/2\rceil$ poisoned data owners. Thus, an adversary poisoning  at most $t = \lceil(c_y-c_{y'})/2\rceil-1$ data owners still generates the same prediction $y$ on $x$.
\end{proof}

\smallskip
\myparagraph{Privacy Analysis}
Recent work by  McMahan et al. \cite{User-levelDP} introduced the notion of  \emph{user-level} differential privacy where the presence of a user in the protocol should have imperceptible impact on the final trained model. We show that, given our threat model, SafeNet provides the strong privacy guarantee of user-level differential privacy, which also implies example-level differential privacy. This privacy guarantee can protect against  model extraction and property inference attacks, in addition to membership inference attacks.


\begin{theorem}
When $\textsc{DPNoise}$ function  samples from a Laplace random variable $Lap(2/\varepsilon)$, Algorithm 2 satisfies user-level $\varepsilon$-differential privacy.
\label{thm:dp}
\end{theorem}

\begin{proof}
Observe that replacing a local model obtained from a data owner in our framework only changes $\textsc{Counts}$ for two classes by 1 on any given query, so it has an $\ell_1$ sensitivity of 2. As a result, $\text{Lap}(2/\varepsilon)$ suffices to ensure that user-level $\varepsilon$-differential privacy holds.
\end{proof}

The main crux of Theorem~\ref{thm:dp} is that no model can influence $\textsc{Counts}$ too much, an observation also  made by PATE~\cite{PATE} and the CaPC~\cite{choquettechoo2021capc} framework, but they only considered example-level differential privacy, protecting against membership inference attacks, but not stronger attacks that user-level differential privacy prevents. This limitation is inherent in PATE, as the central training set is split to train multiple models. However, our stronger analysis holds for SafeNet in the private collaborative learning setting, as we start with pre-existing partitions of benign and poisoned datasets.
We prove Theorem~\ref{thm:dp} by considering Laplace noise, but various improvements to PATE using different mechanisms such as Gaussian noise and other data-dependent approaches~\cite{PATE, papernot2018scalable},
can also be extended to our framework.

\smallskip
\myparagraph{Combining Robustness and Privacy}
Adding differentially private noise prevents Algorithm 2 from returning the exact difference between the top two class-label counts, making it only possible to offer probabilistic robustness guarantees. That is, the returned $t$ is actually a noisy version of the ``true'' $t^*$, where $t^*$ is used to certify correctness. However, for several choices of the $\text{DPNoise}$ function, the exact distribution of the noise is known, making it easy to provide precise probabilistic guarantees similar to those provided by Theorem~\ref{thm:robustpredcorrect}. 
%
For example, if Gaussian noise with scale parameter $\sigma$ is used to guarantee DP, and PredGap returns $t$, then this prediction observed $t$, then we know that the true $t^*$ is larger than $t-k$ with probability $\Phi(k/\sigma)$, where $\Phi$ denotes the Gaussian CDF.

\subsection{Extensions}
\label{sec:extensions}
In addition to providing various guarantees, we offer a number of extensions to our original SafeNet design. 

\smallskip
\myparagraph{Transfer Learning}
A major disadvantage of SafeNet is its slower inference time compared to a traditional PPML framework, requiring to perform a forward pass on all local models in the ensemble. However, for transfer learning scenario, we propose a way where SafeNet runs almost as fast as the traditional framework.
In transfer learning \cite{KSL19,DCLT19}, a pre-trained model $\fm_B$, which is typically trained on a large public dataset, is used as a ``feature extractor'' to improve training on a given target dataset. 
In our setting, all data owners start with a common pre-trained model, and construct their local models by fine tuning $\fm_B$'s last `$l$' layers using their local data. We can then modify the prediction phase of  SafeNet to reduce its inference time and cost considerably. The crucial observation is that all local models differ only in the weights associated to the last $l$ layers. Consequently, given a prediction query, we run $\fm_B$ upto its last $l$ layers and use its output to compute the $l$ layers of all the local models to obtain predictions for majority voting. The detailed description of the modified SafeNet algorithm is given in Appendix~\ref{apndx:TransferLearning}. Note that, this approach achieves the same robustness and privacy guarantees as described in Section \ref{sec:Analysis}, given that $\fm_B$ was originally not tampered with.

\smallskip
\myparagraph{Integration Testing}
While SafeNet can handle settings with non-iid data distributions among data owners, the local models accuracies might be impacted by extreme non-iid settings (we analyze the sensitivity of SafeNet to data imbalance in Section \ref{sec:ExpExtensions}). 
In such cases, SafeNet \emph{fails fast}, allowing the owners to determine whether or not using SafeNet is the right approach for their setting. This is possible because SafeNet's training phase is very cheap, making it possible to quickly evaluate the ensemble's accuracy on the global validation set. If the accuracy is not good enough, the owners can use a different approach, such as a standard MPC training. SafeNet's strong robustness  guarantees and an efficient training phase makes it an appealing first choice for private collaborative learning. 




\smallskip
\myparagraph{Low Resource Owners}
If a data owner does not have sufficient resources to train a model on their data, they cannot participate in the standard SafeNet protocol. In such situations, computationally restricted owners can defer their training to SafeNet, that can use standard MPC training approaches to train their models. Training these models in MPC increases the computational overhead of our approach, but facilitates broader participation.
We provide the details of this modification in Appendix \ref{apndx:TrainCRO} and also run an experiment in Appendix \ref{sec:apndx_extension} to verify that SafeNet remains efficient, while retaining the same robustness and privacy properties.





\subsection{Comparison to Existing Ensemble Strategies}
\label{sec:otherensembles}
Model ensembles have been considered to address adversarial machine learning vulnerabilities in several prior works. Here, we discuss the differences between our analysis and previous ensembling approaches.

\paragraph{Ensembles on a Centralized Training Set}
Several ensemble strategies seek to train a model on a single, centralized training set. This includes using ensembles to prevent poisoning attacks~\cite{jia2021intrinsic, levine2020deep}, as well as to provide differential privacy guarantees~\cite{PATE} or robustness to privacy attacks~\cite{tang2022mitigating}. Due to centralization, none of these techniques can take advantage of the partitioning of datasets. As a result, protection from poisoning is only capable of handling a small number of poisoning examples, whereas our partitioning allows large fractions of the entire dataset to be corrupted. 
PATE, due to data centralization, can only guarantee privacy for individual samples, whereas in our analysis, the \emph{entire dataset} of a given owner can be changed, providing us with \emph{user-level} privacy.

\paragraph{CaPC~\cite{choquettechoo2021capc}}
Chouquette-Choo et al.~\cite{choquettechoo2021capc} propose CaPC, which extends PATE to the MPC collaborative learning setting. Their analysis gives differential privacy guarantees for individual examples. Our approach extends their analysis to a differential privacy guarantee for the entire local training set and model, to provide protection against attacks such as property inference and model extraction. In addition, our approach also provides poisoning robustness guarantees which they cannot, as they allow information to be shared between local training sets.

\paragraph{Cao et al.~\cite{CJG21}}
Recent work by Cao et al.~\cite{CJG21} gave provable poisoning robustness guarantees for federated learning aggregation. They  proposed an ensembling strategy, where, given $m$ data owners, $t$ of which are malicious, they construct an ensemble of $\binom{m}{k}$ global models, where each model is trained on a dataset collected from a set of $k$ clients. Our poisoning robustness argument in Theorem ~\ref{thm:robustpredcorrect} coincides with theirs at $k=1$, a setting they do not  consider as their approach relies on combining client datasets for federated learning. 
Additionally, $k=1$ makes their approach  vulnerable to data reconstruction attacks \cite{BDSSSP21}, an issue SafeNet does not face as the attack directly violates the underlying security guarantee of the MPC.  We experimentally compare both approaches on a federated learning dataset in Section~\ref{sec:FL} and show that our approach outperforms \cite{CJG21}.



\section{Evaluation}
\label{sec:exp}

\subsection{Experimental Setup}    
We build a functional code on top of the MP-SPDZ library~\cite{Keller20}\footnote{https://github.com/data61/MP-SPDZ}  to assess the impact of data poisoning attacks on the training phase of PPML frameworks. We consider four different MPC settings, all available in the MP-SPDZ library:  i) two-party with one semi-honest corruption (2PC) based on \cite{DSZ15,CramerDESX18}; ii) three-party with one semi-honest corruption (3PC) based on  Araki et al.~\cite{AFLNO16} with optimizations by~\cite{MR18,DEK20};  iii) three-party with one malicious corruption based on Dalskov et al.~\cite{DEK21}; and iv) four-party  with one malicious corruption (4PC), also based on~\cite{DEK21}. Note, that both semi-honest and malicious adversaries possess poisoning capability; their roles change only inside the SOC paradigm.

In all the PPML  frameworks, the data owners secret-share their training datasets to the servers and a single ML model is trained on the combined dataset. Typically, real number arithmetic is emulated by using $32$-bit fixed-point representation of fractional numbers. Each fractional number $x \in \Z{\ell}$ is represented as $\lfloor x \cdot 2^f \rceil$, where $\ell$ and $f$  denote the ring size and precision, respectively. We set  $\ell = 64$ and  $f=16$. Probabilistic truncation proposed by Dalskov et al.~\cite{DEK20, DEK21} is applied  after every multiplication.
%
In the MPC library implementation, the sigmoid function  for computing the output probabilities  is replaced with a three-part approximation~\cite{MohasselZ17,CCPS19,DEK21}. In SafeNet, models are trained locally using the original sigmoid function. We implement softmax function using the  method of Aly et al.~\cite{AS19}.
We perform our experiments over a LAN network  on a   $32$-core server with $192$GB of memory allowing up to $20$ threads to be run in parallel.

\subsection{Metrics}
 We use the following metrics to compare  SafeNet  with existing PPML framework:
 
 \smallskip
\myparagraph{Training Time} is the time taken to privately train a model inside the MPC (protocol \pitrain). As is standard practice~\cite{MohasselZ17,MR18,CCPS19,CRS20, BCPS20, DEK21}, this excludes the time taken by the data owners to secret-share their datasets and models to the servers as it is a one-time setup phase. 


\smallskip
\myparagraph{Communication Complexity} is the amount of data exchanged between the  servers during the privacy-preserving execution of the training phase.

\smallskip
\myparagraph{Test Accuracy} is the percentage of test samples that the ML model correctly predicts.

\smallskip
\myparagraph{Attack Success Rate} is the percentage of  target samples that were misclassified as the label of attacker's choice. 

\smallskip
\myparagraph{Robustness against worst-case adversary} We measure the resilience of SafeNet at a certain corruption level $c$ against a powerful, worst-case adversary. For each test sample, this adversary can select any subset of $c$ owners, arbitrarily modifying the model to change the test sample's classification. This is the same adversary considered in Algorithm 2 and by Theorem~\ref{thm:robustpredcorrect}, any any model which is robust against this attack has a provably certified prediction. We measure the error rate on testing samples for this worst-case adversarial model. 

\subsection{Datasets and Models} We give a  descriptions of the datasets and models used for our experiments  below.


\smallskip
\myparagraph{MNIST} The MNIST  dataset \cite{Dua:2019} is a 10 class classification problem which is used to predict digits between $0$ and $9$. We train a logistic regression model for MNIST.

\smallskip
\myparagraph{Adult} The Adult  dataset \cite{Dua:2019} is for a binary classification problem to predict if a person's annual income is above \$50K. We train a neural network with one hidden layer of size $10$ nodes using ReLU activations. 

\smallskip  
\myparagraph{Fashion} We  train several neural networks on the Fashion-MNIST dataset \cite{xiao2017/online}  with one to three hidden layers. The Fashion dataset is a 10-class classification problem  with $784$ features representing various garments. All hidden layers have $128$ nodes and  ReLU activations, except the output layer using softmax.  

\smallskip
\myparagraph{CIFAR-10} The CIFAR-10 dataset~\cite{krizhevsky2009learning} is a 10 class image dataset. CIFAR-10 is harder than other datasets we consider, so we perform transfer learning from a ResNet-50 model~\cite{he2016deep} pretrained on the ImageNet dataset~\cite{deng2009imagenet}. We fine tune only the last layer, freezing all convolutional layers.

\smallskip
\myparagraph{EMNIST} The EMNIST dataset~\cite{EMNIST} is a benchmark federated learning image dataset, split in a non-iid fashion by the person who drew a given image. We select 100 EMNIST clients in our experiments.

\subsection{Dataset Partitioning and Model Accuracy}
We conduct our experiments by varying the number of data owners.
We  split MNIST and Adult datasets across 20 participating data owners, while we use 10 owners for  Fashion and CIFAR-10 datsets. The EMNIST dataset used for comparison with   prior work on federated learning assumes $100$ participating owners. Each owner selects at random $10\%$ of its local training data as the validation dataset $\cv{j}$. All  models are trained using mini-batch stochastic gradient descent.
 
 To introduce non-iid behavior in our datasets (except for EMNIST, which is naturally non-iid), we sample class labels from a Dirichlet distribution~\cite{HSU19}.  
That is, to generate a population of non-identical owners, we sample $q \sim Dir(\alpha p) $  from a Dirichlet distribution, where $p$ characterizes a prior class distribution over all distinct classes, and $\alpha > 0$ is a concentration parameter which controls the degree of similarity between owners.
As $\alpha \rightarrow \infty $, all owners have identical distributions, whereas as $\alpha \rightarrow 0 $, each owner holds samples of only one randomly chosen class. In practice, we observe $\alpha = 1000$ leads to almost iid behavior, while $\alpha = 0.1$ results in an extreme imbalance distribution. 
The default choice for all our experiments is  $\alpha = 10$, which provides a realistic non-iid distribution. We will vary parameter $\alpha$ in Appendix~\ref{sec:apndx_extension}.

{\begin{table}[h!]
		\centering 
		\begin{adjustbox}{max width=0.5\textwidth}{  
				\begin{tabular}{c  c  c  c	 c}
					
					
					Dataset  &  Partition Type & Local Model  &   SafeNet Ensemble & Improvement\\
					
					\midrule
					
					MNIST & \multirow{4}{*}{Dirchlet} & 80.05\% & 89.48\% & 9.03\%\\
					
					Adult &  & 77.32\% & 81.41\% & 4.09\%\\
					
					FASHION & & 71.68\% & 83.26\% & 11.53\%\\
					
					CIFAR-10 &  & 54.03\% & 62.76\% & 8.73\% \\
					
					\midrule
					EMNIST & Natural & 54.05\% & 79.19\% & 25.14\% \\
					

				\end{tabular}
			}
		\end{adjustbox}
		\caption{\small  Test accuracy comparison of a single local model and the entire SafeNet ensemble. SafeNet Ensemble improves upon a single local model across all datasets.}
		\label{tab:LocalAcc}
	\end{table}
	
}

We  measure the accuracy of a local model trained by individual data owners and our SafeNet ensemble. Table \ref{tab:LocalAcc} provides the detailed comparison of the accuracy of the local and ensemble models across all four datasets. We observe that SafeNet consistently outperforms local models, with improvements ranging from 4.09\% to 25.14\%. The lowest performance is on CIFAR-10, but in this case SafeNet's accuracy is very close to fine-tuning the network using the combined dataset, which reaches 65\% accuracy.


\subsection{Implementation of Poisoning Attacks}

\smallskip
\myparagraph{Backdoor Attacks} 
\label{sec:log_bd}
We use the  BadNets attack by Gu et al. \cite{GLDG19}, in which the poisoned owners inject a backdoor into the model to change the model's  prediction from source label $y_s$ to target label $y_t$. For instance, in an image dataset, a backdoor might set a few pixels in the corner of the image to white. The BadNets attack strategy simply identifies a set of $k$ target samples $\lbrace x^t_i\rbrace_{i=1}^k$ with true label $y_s$, and creates backdoored samples   with target label $y_t$. We use $k = 100$ samples, which is sufficient to poison all  models.

 

To run backdoor attacks on models trained with standard PPML frameworks, the poisoned owners create the poisoned dataset $\ld^*_j$ by adding  $k$ poisoned samples and secret-sharing them as part of the training dataset to the MPC. The framework then trains the ML model on the combined dataset submitted by both the honest and poisoned owners.


In SafeNet, backdoor attacks are implemented at  the poisoned owners, which add $k$ backdoored samples to their  dataset $\ld_j$   and train their local models $\fm^*_j$ on the combined clean and poisoned data. A model trained only on poisoned data will be easy to filter due to low accuracy, making training on clean samples necessary.  The corrupt owners then secret-share both the model $\fm^*_j$ and validation set $\cv{j}$ selected at random from $\ld_j$ to the MPC.




\smallskip
\myparagraph{Targeted Attacks} 
\label{sec:log_tgt} We select $k$ targeted samples, and change their labels in training to a target label $y_t$ different from the original label. 
The models are trained to simultaneously minimize both the training and the adversarial loss. This strategy has also been used to construct poisoned models by prior work~\cite{koh2018stronger}, and can be viewed as an unrestricted version of the state-of-the-art Witches' Brew targeted attack (which requires clean-label poisoned samples)~\cite{geiping2020witches}. 

The next question to address is which samples to target as part of the attack. We use two strategies to generate $k=100$ target samples, based on an ML model trained by the adversary over the test data. In the first strategy, called TGT-Top, the adversary chooses examples classified correctly with high confidence by a different model. Because these examples are easy to classify, poisoning them should be hard. We also consider an attack  called TGT-Foot, which chooses low confidence examples, which are easier to poison. For both strategies, the adversary replaces its label with the second highest predicted label. We compare these two strategies for target selection.

The difference between targeted and backdoor attacks is that  targeted attacks do not require the addition of a backdoor trigger  to  training or testing samples, as needed in a backdoor attack. However, the impact of the backdoor attack is larger. Targeted attacks change the prediction on a small set of testing samples (which are selected in advance before training the model), while the backdoor attack generalizes to any testing samples including the backdoor pattern.

\subsection{Evaluation on Logistic Regression} \label{sec:LogReg}



We start with  DIGIT 1/7 dataset, a subset of MNIST data using only digits 1 and 7, for which we evaluate the computational costs and the poisoning attack success, for both traditional PPML and our newly proposed SafeNet framework. 

We perform our experiments over four underlying MPC frameworks, with both semi-honest and malicious adversaries. Table \ref{tab:TC_BD} provides a detailed analysis of the training time and communication complexity for both existing PPML and SafeNet frameworks. Note that the  training time and communication cost for the PPML frameworks is reported per epoch times the number of epochs in training. The number of epochs is a configurable hyper-parameter, but usually at least 10 epochs are required.  On the other hand, the training time and communication reported for SafeNet  is for the end-to-end execution inside the MPC, independent of the number of epochs. We observe large improvements of SafeNet over the existing PPML frameworks. For instance, in the semi-honest two-party setting, SafeNet achieves $30\times$ and $17\times$ improvement in running time and communication complexity, respectively, for $n=10$ epochs. This is expected because SafeNet performs local model training, which is an expensive phase in the MPC.




{\begin{table}[h!]
		\begin{adjustbox}{max width=0.45\textwidth}{  
				\begin{tabular}{c  c  c	 r  r}
					
					
					MPC  &  Setting  &   Framework & Training (s) & Comm. (GB)\\
					
					\midrule
					
					\multirow{2}{*}{2PC} & Semi-Honest & PPML & n$\times$151.84 & n$\times$65.64\\
					
					&\cite{DSZ15}	& SafeNet & $57.41$ & $38.03$\\
					
					\midrule
					
					\multirow{4}{*}{3PC} & Semi-Honest & PPML & n$\times$2.63 & n$\times$0.35 \\
					
					& \cite{AFLNO16}	& SafeNet & $0.54$ & $0.15$\\
					
					\cmidrule(lr){2-5}
					
					& Malicious & PPML & n$\times$32.54 & n$\times$ 2.32 \\
					
					&\cite{DEK21}	& SafeNet & $9.44$ & $1.47$\\		
					
					\midrule
					
					\multirow{2}{*}{4PC} & Malicious  & PPML & n$\times$5.28 & n$\times$0.66\\
					
					& \cite{DEK21} & SafeNet & $1.09$ & $0.28$\\
					
					
				\end{tabular}
			}
		\end{adjustbox}
		\caption{\small  Training Time (in seconds) and Communication (in GB) of existing PPML and SafeNet framework for a logistic regression model over several MPC settings over a LAN network. n denotes the number of epochs required for training the logistic regression model in the PPML framework. The time and communication reported for SafeNet is for end-to-end execution.}
		\label{tab:TC_BD}
		
		
	\end{table}
	
}


To mount the backdoor attack, the backdoor pattern sets the top left pixel value to white (a value of 1). We set the original class as $y_s=1$ and target class as $y_t=7$. Figure~\ref{fig:Logreg_ASR} (a) shows the success rate for the 3PC PPML  and SafeNet frameworks by varying the number of poisoned owners between 0 and 10.   We  tested with all four PPML settings and the results are similar. We observe that by poisoning data of a single owner, the adversary is successfully able to introduce a backdoor in the PPML framework. The model in the PPML framework predicts {all} $k=100$ target samples as $y_t$, achieving $100\%$ adversarial success rate. In contrast,  SafeNet is successfully able to defend against the backdoor attack, and provides $0\%$ attack success rate up to 9 owners with poisoned data. The test accuracy on clean data for both frameworks is high at around $98.98\%$ even after increasing the poisoned owners to $10$. 

 
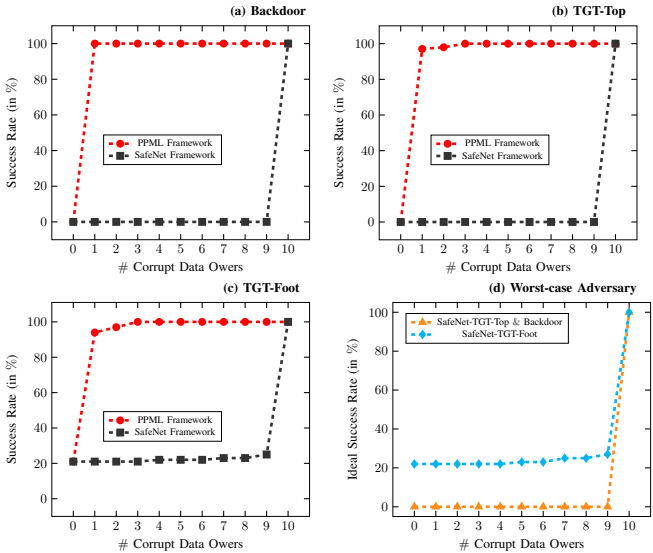
\begin{figure}[htb!]
	\begin{minipage}{.24\textwidth}
		\scalebox{.5}{
			
			\begin{tikzpicture}
				\begin{axis}[legend style={at={(0.2,0.35)},anchor=south west, nodes={scale=0.75, transform shape}}, xlabel={$\#$ Corrupt Data Owers}, ylabel={Success Rate (in $\%$)}, xtick = data]

					\addplot[color = red!805,mark=*,dashed, every mark/.append style={solid}, line width = 2] coordinates { (0,0) (1, 100) (2, 100) (3, 100) (4,100) (5, 100) (6,100) (7,100) (8,100) (9,100) (10,100)};
					\addlegendentry{PPML Framework}

					\addplot[color = black!80,mark=square*,dashed, every mark/.append style={solid}, line width = 2] plot coordinates {(0,0) (1, 0) (2, 0) (3,0) (4,0) (5,0) (6,0) (7,0) (8,0) (9,0) (10,100)};
					\addlegendentry{SafeNet Framework}
					

				\end{axis}
				\node[align=right,font=\bfseries, xshift=8.5em, yshift=1.0em] (title) at (current bounding box.north) {(a) Backdoor};
			\end{tikzpicture}
		}
	\end{minipage}%
	\begin{minipage}{.24\textwidth}
		\scalebox{.5}{
			\begin{tikzpicture}
				\begin{axis}[legend style={at={(0.2,0.35)},anchor=south west, nodes={scale=0.75, transform shape}}, xlabel={$\#$ Corrupt Data Owers}, ylabel={Success Rate (in $\%$)}, xtick = data, ymin = -10, ymax = 110]

					\addplot[color = red!805,mark=*,dashed, every mark/.append style={solid}, line width = 2] coordinates { (0,0) (1,97.04) (2,98) (3,100) (4,100) (5,100) (6,100) (7,100) (8,100) (9,100) (10,100)};
					\addlegendentry{PPML Framework}

					\addplot[color = black!80,mark=square*,dashed, every mark/.append style={solid}, line width = 2] plot coordinates {(0,0) (1,0) (2,0) (3,0) (4,0) (5,0) (6,0) (7,0) (8,0) (9,0) (10,100)};
					\addlegendentry{SafeNet Framework}
					

				\end{axis}
				\node[align=right,font=\bfseries, xshift=8em, yshift=1.0em] (title) at (current bounding box.north) {(b) TGT-Top};
			\end{tikzpicture}
		}
	\end{minipage}%

	\begin{minipage}{.25\textwidth}
		\scalebox{.5}{
			\begin{tikzpicture}
				\begin{axis}[legend style={at={(0.2,0.35)},anchor=south west, nodes={scale=0.75, transform shape}}, xlabel={$\#$ Corrupt Data Owers}, ylabel={Success Rate (in $\%$)}, xtick = data, ymin = -10]

					\addplot[color = red!805,mark=*,dashed, every mark/.append style={solid}, line width = 2] coordinates { (0,21) (1,94.00) (2,97.00) (3,100) (4,100) (5,100) (6,100) (7,100) (8,100) (9,100) (10,100)};
					\addlegendentry{PPML Framework}

					\addplot[color = black!80,mark=square*,dashed, every mark/.append style={solid}, line width = 2] plot coordinates {(0,21) (1,21) (2,21) (3,21) (4,22) (5,22) (6,22) (7,23) (8,23) (9,25) (10,100)};
					\addlegendentry{SafeNet Framework}
					

				\end{axis}
				\node[align=right,font=\bfseries, xshift=8em, yshift=1.0em] (title) at (current bounding box.north) {(c) TGT-Foot};
			\end{tikzpicture}
		}
	\end{minipage}%
	\begin{minipage}{.25\textwidth}
		\scalebox{0.5}{
			\begin{tikzpicture}
				\begin{axis}[legend style={at={(0.05,0.95)},anchor=north west , nodes={scale=0.75, transform shape}}, xlabel={$\#$  Corrupt Data Owers}, ylabel={ Ideal Success Rate (in $\%$)}, ymin = -5, ymax = 105, xtick = data]
					
					\addplot[color = orange!90,mark=triangle*,dashed,every mark/.append style={solid}, line width = 2] coordinates { (0,0) (1,0) (2, 0) (3,0) (4,0) (5,0) (6,0) (7,0) (8,0)(9,0)(10,100)};
					\addlegendentry{SafeNet-TGT-Top $\&$ Backdoor}

					\addplot[color = cyan!90,mark=diamond*,dashed,every mark/.append style={solid}, line width = 2] plot coordinates {(0,22) (1,22) (2,22) (3,22) (4,22) (5,23) (6,23) (7,25) (8,25) (9,27) (10,100)};
					\addlegendentry{SafeNet-TGT-Foot}
					
					


				\end{axis}
				\node[align=right,font=\bfseries, xshift=5.0em, yshift=1.0em] (title) at (current bounding box.north) {(d) Worst-case Adversary};
			\end{tikzpicture}
		}
	\end{minipage}%
	\caption{\small  Logistic regression attack success rate on the Digit-1/7 dataset for PPML and SafeNet frameworks in the 3PC setting, for varying poisoned owners launching Backdoor and Targeted attacks. Plot (a) gives the success rate for the BadNets attack, while  plots (b) and (c) show the success rates for  the TGT-Top and TGT-Foot targeted attacks. Plot (d) provides the worst-case adversarial success when the set of poisoned owners can change per sample. Lower attack success result in increased robustness. SafeNet achieves much higher level of robustness than existing PPML under both attacks.} \label{fig:Logreg_ASR}
	\vspace{-2mm}
\end{figure}

We observe in Figure~\ref{fig:Logreg_ASR} (b) that for the TGT-Top targeted attack, a single owner poisoning is able to successfully misclassify $98\%$ of the target samples in the PPML framework. As a consequence, the test accuracy of the model drops by $\approx 10\%$. In contrast,  SafeNet  works as intended even at high levels of poisoning. For the TGT-Foot attack in Figure~\ref{fig:Logreg_ASR} (c), the test accuracy of the 3PC PPML framework drops by $\approx 5 \%$. The attack success rate is $94\%$ for the 3PC PPML, which is  decreased to $21\%$ by SafeNet, in presence of a single poisoned owner. The accuracy drop and success rate vary  across the two strategies because of the choice of the target samples. In TGT-Foot, the models have low confidence on the target samples, which introduces errors even without poisoning, making the attack succeed with slightly higher rate in SafeNet. Still, SafeNet provides resilience against both TGT-Top and TGT-Foot for up to 9 out of 20 poisoned owners. 



\smallskip
\myparagraph{Worst-case Robustness}
Figure~\ref{fig:Logreg_ASR} (d) shows the worst-case attack success in SafeNet, by varying the number of poisoned owners $c \in [1,10]$ and allowing the attacker to poison a different set of $c$ owners for each testing sample (i.e., the adversarial model considered in Algorithm 2 for which we can certify predictions). Interestingly,  SafeNet's accuracy is similar to that achieved under our backdoor and targeted attacks, even for this worst-case adversarial scenario. Based on these results we conclude that: (1) the backdoor and targeted attacks we choose to implement are as strong as the worst-case adversarial attack, in which the set of poisoned owners is selected per sample; (2) SafeNet provides certified robustness up to 9 out of 20 poisoned owners even under this powerful threat scenario. 

\smallskip
\myparagraph{Multiclass Classification} We  also test both frameworks in the multiclass classification setting for both Backdoor and Targeted attacks on MNIST dataset and  observe similar large improvements. For instance, in the semi-honest 3PC setting, we get $240\times$ and $268\times$ improvement, respectively, in training running time and communication complexity for $n=10$ epochs while the success rate in the worst-case adversarial scenario not exceeding $50\%$ with $9$ out of $20$ owners being poisoned. This  experiment shows that the robust accuracy property of our framework translates seamlessly even for the case of a multi-class classification problem.  The details of the experiment  are deferred to Appendix  \ref{app:bench}.

\subsection{Evaluation on Deep Learning Models} \label{sec:DNN}

We evaluate  neural network training for  PPML and SafeNet frameworks on the Adult and Fashion datasets. We provide experiments on a three hidden layer neural network on Fashion in this section and include additional experiments in Appendix~\ref{app:bench}.

 {\begin{table*}[h!]
		\centering \scriptsize
		\begin{adjustbox}{max width=\textwidth}{  
				\begin{tabular}{c   c	 c  r  r  r  r  r  r r}
					
					
					\multirow{2}{*}{MPC } & \multirow{2}{*}{ Setting} &  \multirow{2}{*}{Framework} &  \multicolumn{1}{c}{\multirow{2}{*}{Training Time (s)}} & \multicolumn{1}{c}{\multirow{2}{*}{Communication (GB)}} & \multicolumn{2}{c}{Backdoor Attack} & \multicolumn{3}{c}{Targeted Attack}\\ \cmidrule{6-7} \cmidrule{7-10}
					
					&     &  &  &  & Test Accuracy & Success Rate  & Test Accuracy & Success Rate-Top & Success Rate-Foot\\
					
					\midrule
					
					
					\multirow{2}{*}{3PC \cite{AFLNO16}} & \multirow{2}{*}{Semi-Honest}  & PPML & n $\times~565.45$ & n $\times~154.79$ & $84.07\%$ & $100\%$ & $82.27\%$ & $100\%$ & $100\%$\\
					
					&  & SafeNet & $156.53$ & $41.39$ & $84.36\%$ & $0\%$ & $84.48\%$ & $0\%$ & $32\%$\\
					
					
					
					
					
					
					
					\midrule
					
					\multirow{2}{*}{4PC \cite{DEK21}} & \multirow{2}{*}{Malicious} & PPML & n $\times~1392.46$ & n $\times~280.32$ & $84.12\%$ & $100\%$ & $82.34\%$ & $100\%$ & $100\%$\\
					
					&  & SafeNet & $356.26$ & $76.43$ & $84.36\%$ & $0\%$ & $84.54\%$ & $0\%$ & $32\%$\\
					
					
					
					
					
					

				\end{tabular}
			}
		\end{adjustbox}
        \caption{\small  Time (in seconds) and Communication (in Giga-Bytes) over a LAN network for PPML and SafeNet framework  training a Neural Network model with 3 hidden layers over Fashion dataset. n denotes the number of epochs used to train the NN model in the PPML framework. The time and communication reported for SafeNet is for end-to-end execution. Test Accuracy and Success Rate is given for the case when a single owner is corrupt.}
		\label{tab:Fashion_nn}
		\vspace{-2mm}
		
	\end{table*}
	
}

Table~\ref{tab:Fashion_nn} provides a detailed analysis of the training time, communication, test accuracy and success rate for the 4PC PPML  framework and SafeNet using one poisoned owner. 
We observe that SafeNet has $39\times$ and $36\times$ improvement in training time and communication complexity over the PPML framework, for $n=10$ epochs. The SafeNet prediction time is on average $26$ milliseconds to perform a single secure prediction, while the existing PPML framework takes on average $3.5$ milliseconds for the same task. We believe this is a reasonable cost for many applications, as SafeNet  has significant training time improvements and robustness guarantees. 

For the BadNets backdoor attack we set the true label $y_s$ as a `T-Shirt' and target label $y_t$ as `Trouser'.  We test the effect of both TGT-Top and TGT-Foot attacks  under multiple poisoned owners, and also evaluate another variant of targeted attack called TGT-Random, where we randomly sample $k=100$ target samples from the test data. Figure~\ref{fig:RA_FASHION} provides the worst-case adversarial success of SafeNet  against these attacks. We observe that SafeNet provides certified robustness for TGT-Random and TGT-Top up to 4 out of 10 poisoned onwers, while the adversary is able to misclassify more target samples in the TGT-Foot attack. The reason is that the $k$ selected target samples have lowest confidence and models in the ensemble are likely to be in disagreement on their prediction.   


\begin{figure}[htb!]
	\centering
	\begin{minipage}{.3\textwidth}
		\scalebox{0.6}{
			\begin{tikzpicture}
				\begin{axis}[legend style={at={(0.05,0.95)},anchor=north west , nodes={scale=0.7, transform shape}}, xlabel={$\#$ Corrupt Data Owers}, ylabel={Ideal Success Rate (in $\%$)}, ymin = -5, ymax = 105, xtick = data,height=6.5cm,width=8.5cm]
					
					\addplot[color = red,mark=*,dashed,every mark/.append style={solid}, line width = 2] coordinates { (0,0) (1,0) (2, 0) (3,0) (4,0) (5,100)};
					\addlegendentry{SafeNet-TGT-Top}
					
					\addplot[color = black!80,mark=triangle*,dashed, every mark/.append style={solid}, line width = 2] plot coordinates {(0,6) (1,8) (2, 12) (3,19) (4,31) (5,100)};
					\addlegendentry{SafeNet-TGT-Random}
					
					\addplot[color = orange!90,mark=diamond*,dashed,every mark/.append style={solid}, line width = 2] plot coordinates {(0,33) (1,36) (2,46) (3,71) (4,91) (5,100)};
					\addlegendentry{SafeNet-TGT-Foot}
					
					\addplot[color = cyan!90, mark=square*,dashed, every mark/.append style={solid}, line width = 2] plot coordinates {(0,0) (1,0) (2,1) (3,7) (4,36) (5,100)};
					\addlegendentry{SafeNet-Backdoor}


				\end{axis}
			\end{tikzpicture}
		}
	\end{minipage}%
	\caption{\small Worst-case adversarial success against targeted and backdoor attacks of a three-layer neural network trained on Fashion  in  SafeNet. The adversary can change the set of $c$ poisoned owners per sample. SafeNet achieves robustness on the backdoor, TGT-Top and TGT-Random attacks, up to 4 poisoned owners out of 10. The TGT-Foot attack targeting low-confidence samples has higher success.} \label{fig:RA_FASHION}
	\vspace{-2mm}
\end{figure}
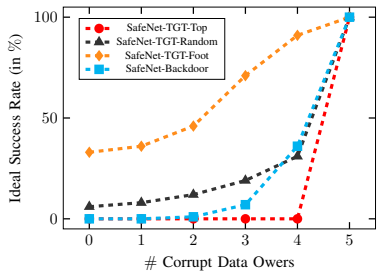

 \label{sec:fashion_nn}

%

\subsection{Evaluation of Extensions} \label{sec:ExpExtensions}

Here, we evaluate  our SafeNet extensions introduced in Section~\ref{sec:extensions}. First, we experiment with our transfer learning extension. We show that, on applying our extension  to SafeNet,  its inference overhead falls dramatically. We test our approach on Fashion and CIFAR-10 datasets.
%
%
For the Fashion dataset, we use the same setup as earlier with $m = 10$ data owners, and three-layered neural network as the model architecture, where each data owner fine-tunes only the last layer ($l=1$) of the pre-trained model. We observe that for  each  secure inference, SafeNet   is now only $1.62\times$ slower and communicates $1.26\times$ more  on average than the PPML framework, while the standard SafeNet approach is about $8\times$ slower due to the evaluation of multiple ML models. 

We observe even better improvements for CIFAR-10 dataset. Here, we use a state-of-the-art 3PC inference protocol from \cite{KRCGRS20}, built specially for ResNet models. In our setting, each owner fine-tunes the last layer of a ResNet-50 model, which was pre-trained on ImageNet data. SafeNet reaches 62.8\% accuracy, decaying smoothly in the presence of poisoning: 51.9\% accuracy tolerating a single poisoned owner, and 39.8\% while tolerating two poisoned owners.
The cost of inference for a single model is an average of 59.9s, and SafeNet's overhead is negligible (experimental noise has a larger impact than SafeNet); SafeNet increases communication by only 0.1\%, increasing around 7MB over the 6.5GB required for standard inference. 

Next, we analyze the behavior of SafeNet under different non-iid settings by varying the concentration parameter $\alpha$. We use the same Fashion dataset setup from Section~\ref{sec:DNN}. We observe that as $\alpha$ decreases, i.e., the underlying data distribution of the owners become more non-iid, SafeNet's accuracy decreases, as expected, but SafeNet still achieves reasonable robustness even under high data imbalance (e.g., $\alpha = 1$). In extremely imbalanced settings, such as $\alpha=0.1$, SafeNet can identify low accuracy during training and data owners can take actions accordingly. We defer the details for this  extension  to Appendix \ref{sec:apndx_extension}, which also includes analyzing attack success rates under extreme non-iid conditions.








\section{Discussion and Comparison} 
\label{sec:LimDist}

We showed that SafeNet successfully mitigates a variety of data poisoning attacks. We now discuss other aspects of our framework such as  scalability and modularity, parameter selection in practice and comparison against other mitigation strategies and federated learning approaches.

\subsection{SafeNet's Scalability and Modularity}
\smallskip
\myparagraph{Scalability} The training and prediction times of  SafeNet  inside the MPC depend on the number of models in the ensemble and the size of the validation dataset. 
The training time increases linearly with  the fraction of training data used for validation and the number of models in the ensemble. Similarly, the prediction phase of SafeNet has both runtime and communication scaling linearly with the number of models in the ensemble. However, we discussed how transfer learning can reduce the inference time of SafeNet.


\smallskip
\myparagraph{Modularity} Another key advantage of SafeNet is that it can use any MPC protocol as a backend, as long as it implements standard  ML operations. We demonstrated this by performing experiments with both malicious and semi-honest security for four different MPC settings. As a consequence, advances in ML inference with MPC will improve SafeNet's runtime. SafeNet can also use any model type implementable in MPC; if more accurate models are designed, this will lead to improved robustness and accuracy. 

\subsection{Instantiating SafeNet in Practice} \label{subsec:Params}
In this section we discuss how SafeNet can be instantiated in practice. There are two aspects the data owners need to agree upon before instantiating SafeNet: i) The MPC framework used for secure training and prediction phase and ii) the parameters in Theorem \ref{thm:LB} to achieve  poisoning robustness. The MPC framework is agreed upon  by choosing the total number of outsourced servers $N$ participating in the MPC, the number of corrupted servers $T$ and the nature of the adversary (semi-honest or malicious in the SOC paradigm). The owners then agree upon a filtering threshold $\threshold$ and the number of poisoned owners $t$ that can be tolerated. Once these parameters are chosen the maximum allowed error probability of the local models trained by the  honest owners based on Lemma \ref{lem:VS} and Theorem \ref{thm:LB}, can be computed as $p < \min( \frac{m(1- \threshold)-t}{m-t}, \frac{m-2t}{2(m-t)})$, where $m$ denotes the total number of data owners. Given the upper bound on the error probability $p$, each honest owner trains its local model while satisfying the above constraint.

We provide a concrete example on parameter selection as follows: We instantiate our Fashion dataset setup, with $m = 10$ data owners participating in SafeNet. For  the MPC framework we choose a three-party setting ($N = 3$ servers), tolerating $T = 1$ corruption.  For poisoning robustness, we set $\threshold = 0.3$ and the number of poisoned owners to $t = 2$.  This gives us the upper bound on max error probability  as $ p < 0.375$. Also the size of the global validation dataset is $|\CV| > 92$ samples, i.e., each data owner contributes $10$ cross-validation samples each such that the constrained is satisfied. With this instantiation, we observe that none of the clean models are filtered during training and the attack success rate of the adversary for  backdoor attacks remains the same even after poisoning $3$ owners, while our analysis holds for $t=2$ poisoned owners. Thus, in practice  SafeNet is able  tolerate more poisoning than our analysis suggests.

\subsection{Comparing to poisoning defenses} \label{sec: otherDefenses}
Defending against poisoning attacks is an active area of research, but defenses tend to be heuristic and specific to attacks or domains. Many defenses for backdoor poisoning attacks exist \cite{LDG18, TLM18,  CCBLELMS19, WYSLVZZ19}, but these strategies work only for Convolutional Neural Networks trained on image datasets; Severi et al.~\cite{SMCO21} showed that these approaches fail when tested on other data modalities and models. Furthermore, recent work by Goldwasser et.al \cite{GKVZ22} formulated a way to plant backdoors that are undetectable by any defense.  In contrast, SafeNet is model agnostic and works for a variety of data modalities. Even if an attack is undetectable, the adversary can poison only a subset of models, making the ensemble robust against poisoning. 
%
%
In certain instances SafeNet can tolerate around $30\%$ of the training data being poisoned, while being attack agnostic. SafeNet is also robust to stronger model poisoning attacks~\cite{BVES18, BCMC19, FCJG20}, which are possible when data owners train their models locally. SafeNet tolerates model poisoning because each model only contributes to a single vote towards the final ensemble prediction. In fact, all our empirical and theoretical analysis of SafeNet is computed for arbitrarily corrupted models. 


\subsection{Comparison with Federated Learning} \label{sec:FL}
Federated Learning (FL) is  a distributed machine learning framework that allows clients to train a  global model without sharing their local training datasets to the central server. However, it differs from the PPML setting we consider in the following ways: (1) Clients do not share their local data to the server in FL, whereas PPML allows sharing of datasets; (2) Clients participate in multiple rounds of training in FL, whereas they communicate only once with the servers in PPML; (3) Clients receive the global model at each round in FL, while in SafeNet they secret-share their models once at the start of the protocol; and, finally, (4) PPML provides stronger confidentiality guarantees such as privacy of the global model. 


It is possible to combine FL and MPC to guarantee both client and global model privacy \cite{KLSYYGCWAW20,HGSN20, FZXGWZ20}, but this involves large communication overhead and is susceptible to poisoning \cite{LGYFLZ22}. For example, recent work \cite{XKG19,BCMC19,BVHES20}  showed that malicious data owners can significantly reduce the learned global model's accuracy. Existing defenses against such owners use Byzantine-robust aggregation rules such as trimmed mean \cite{YCRB19}, coordinate-wise mean \cite{YCRB18} and  Krum \cite{BMGS17}, which have been show to be susceptible to backdoor and model poisoning attacks \cite{FCJG20}. Recent work in FL such as FLTrust \cite{CFLG21} and DeepSight \cite{RNMS22} provide mitigation against backdoor attacks. Both strategies are inherently heuristic, while SafeNet offers provable robustness guarantees. FLTrust also requires access to a clean  dataset, which is not required in our framework, and DeepSight inspects each model update before aggregation, which is both difficult in MPC and leads to privacy leakage from the updates, a drawback not found in SafeNet. An important privacy challenge is that federated learning approaches permit data reconstruction attacks when the central server is malicious \cite{BDSSSP21}. SafeNet prevents such an attack, as it directly violates the security guarantee of the MPC, when instantiated for the malicious setting.


We experimentally compare SafeNet to the federated learning-based approach of Cao et al.~\cite{CJG21}, who also gave provable robustness guarantees in the federated averaging scenario. We instantiate  their strategy for EMNIST dataset and compare their Certified Accuracy metric to SafeNet's, with $m = 100$ data owners, $k = \{2,4\}$ and FedAvg as the base algorithm. To ensure both approaches have similar inference times, we fix the ensemble size to 100 models, each trained using federated learning with 50 global and local iterations.

\begin{figure}[h!]
    \centering
    
    \includegraphics[width=0.3\textwidth]{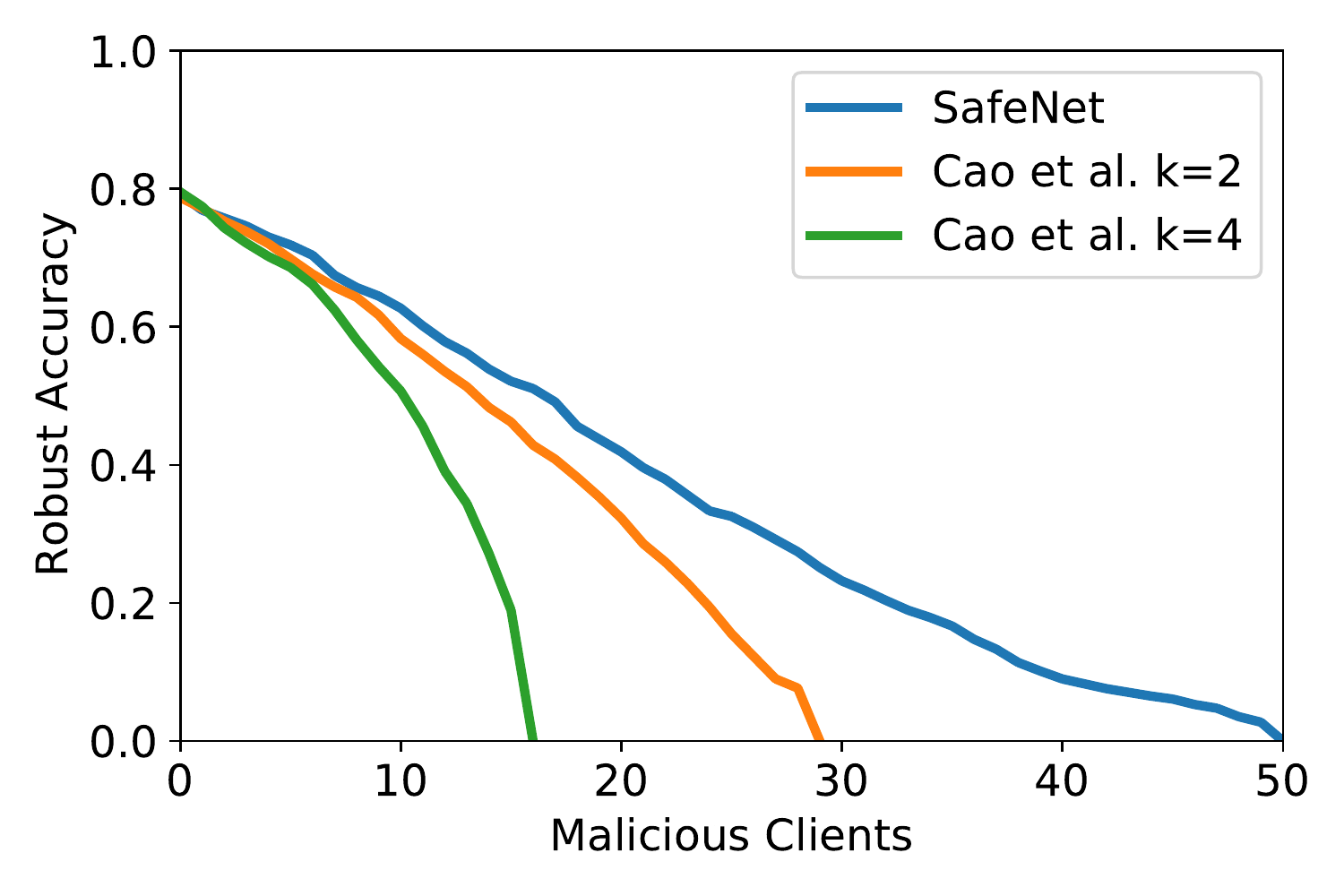}
    \caption{Certified Accuracy of our framework compared to Cao et al.~\cite{CJG21}. We fix the size of the Cao et al. ensemble to 100, to match the test runtime of SafeNet.}
    \label{fig:caopcomarison}
    \vspace{-2mm}
\end{figure}

Figure \ref{fig:caopcomarison} shows that SafeNet consistently outperforms \cite{CJG21}, in terms of maintaining a high certified accuracy in the presence of large poisoning rates. Moreover, their strategy is also particularly expensive at training time when instantiated in MPC. During training, their approach requires data owners to interact inside MPC to train models over multiple rounds. By contrast, SafeNet only requires interaction in MPC at the beginning of the training phase, making it significantly faster. 



\section{Conclusion}

In this paper, we extend the security definitions of  MPC to account for data poisoning attacks when training machine learning models privately. We consider a novel adversarial model who can manipulate the training data of a subset of owners and control a subset of servers in the MPC. We then propose SafeNet, which performs ensembling in MPC, and show that our design has provable  robustness and privacy guarantees, beyond those offered by existing approaches. We evaluate SafeNet  using logistic regression and neural networks models trained on five datasets by varying the distribution similarity across data owners. We consider both end-to-end and transfer learning scenarios. We demonstrate experimentally that SafeNet achieves even higher robustness than its theoretical analysis against backdoor and targeted poisoning attacks, at a significant performance improvement in the training time and communication complexity compared to existing PPML frameworks.   


\section{Acknowledgments}
We thank Nicolas Papernot and Peter Rindal for helpful discussions and feedback. This research was sponsored by the U.S. Army Combat Capabilities Development Command Army Research Laboratory under Cooperative Agreement Number W911NF-13-2-0045 (ARL Cyber Security CRA). The views and conclusions contained in this document are those of the authors and should not be interpreted as representing the official policies, either expressed or implied, of the Combat Capabilities Development Command Army Research Laboratory or the U.S. Government. The U.S. Government is authorized to reproduce and distribute reprints for Government purposes notwithstanding any copyright notation here on.




\bibliographystyle{plain}
\bibliography{main_short}

\appendices

\section{SafeNet Analysis} 
\label{sec:SafeNetAnalysisProof}


In this section we first provide a detailed proof on the size of the validation dataset $\CV$ such that all clean models clear the filtering stage of the training phase of our framework. We then provide a proof on achieving lower bounds on the test accuracy of our framework given all clean models are a part of the ensemble.

The main idea of deriving the minimum size of $\CV$ uses the point that the errors made by a clean model on a clean subset of samples in $\CV$ can be viewed as a Binomial distribution in $(m-t)n$ and $p$, where $n$ denotes the size of the validation dataset $\cv{k}$ contributed by an owner $\usr_k$. We can then upper bound the total errors made by a clean model by applying Chernoff bound and consequently compute the size of $\CV$.

\begin{restatable}{lemma}{VS}
	\label{lem:VS} 	 Let  $\Apsc$ be an adversary who poisons $t$ out of $m$ data owners and corrupts $T$ out of $N$ servers,  and thus contributes $t$ poisoned models to  ensemble $E$, given as output by Algorithm 1.
	Assume that  $\pistr$ securely realizes functionality $\FpTrain$ and every clean model in $E$ makes an error on a clean sample  with probability at most  $p < 1- \threshold$,  where $\threshold$ is the filtering threshold. 
	
	If the validation dataset has at least $\frac{(2+\delta)m\log1/\epsilon}{\delta^2(m-t)p}$ samples and $0 \leq t<\frac{m(1-\threshold-p)}{(1-p)}$, then all clean models pass the filtering stage of the training phase with probability at least $1-\epsilon$, where  $\delta = \frac{(1-\threshold)m -t}{(m-t)p}-1$ and $\epsilon$ denotes the failure probability.  
\end{restatable}


\begin{proof} \label{proof:VS}
	Assume that each owner contributes equal size validation dataset $\cv{k}$ of $n$ samples, then the combined validation set $\CV$ collected from $m$ data owners is comprised  of  $mn$  i.i.d.  samples. However, given  an adversary $\Apsc$ from our threat model, there can be at most $t$ poisoned owners contributing $tn$ poisoned samples to $\CV$. 
%
We define a Bernoulli random variable as follows:
		\begin{align*}
		X_i = 
		\begin{cases}
			1, & \text{w.p.}~p  \\
			0, & \text{w.p.}~ 1-p
		\end{cases}
	\end{align*}

	where $X_i$ denotes if a clean model makes an error on the $i^{th}$ clean sample in the  validation dataset.  Then there are $ \text{Bin}((m-t)n, p)$ errors made by the clean model on the clean subset of samples in $\CV$. Note that, a model passes the filtering stage only when it makes $\geq \threshold mn$ correct predictions. We assume that the worst case where the clean model makes incorrect predictions on all the $tn$ poisoned samples present in $\CV$. As a result, the clean model must make at most $(1-\threshold)mn-tn$ errors on the clean subset of $\CV$ with probability $1-\epsilon$. We can upper bound the probability the model makes at least $(1-\threshold)mn +1 - tn$ errors with a multiplicative Chernoff bound with $\delta > 0$:\newline
	
	\noindent
	\scalebox{0.8}{ $\Prob[\sum_{i=1}^{(m-t)n} X_i > (1-\threshold) mn-tn ] \\=   \Prob\left[\sum_{i=1}^{n} X_i  > (1 + \delta)\mu\right]   < e^{-\frac{\delta^2\mu}{2+\delta}}$}\newline

	where $\mu = (m-t)np$ (the mean of $Bin(mn-tn, p)$) and $\delta = \frac{(1-\threshold)m-t}{(m-t)p}$. The chernoff bound gives that the probability the clean model makes too many errors is at most $e^{-\frac{\delta^2\mu}{2+\delta}} = \epsilon$. Then it suffices to have this many samples: 
	
	$$|\CV| = mn = \frac{(2+\delta)m\log1/\epsilon}{\delta^2(m-t)p}$$
	
	where $\epsilon$ denotes the failure probability and $t < \frac{m(1-\threshold-p)}{(1-p)}$. The inequality on $t$ comes from requiring $\delta>0$.

\end{proof}

As a visual interpretation of  Lemma~\ref{lem:VS},  Figure~\ref{fig:cv_plots} shows the minimum number of samples required in the global validation dataset for varying  number of poisoned owners $t$ and  error probability $p$. We set the total models $m = 20$, the failure probability $\epsilon = 0.01$ and the filtering threshold $\threshold = 0.3$. The higher the values of $t$ and $p$,  the more samples  are required in the validation set. For instance, for $p = 0.20$ and number of poisoned owners $t = 8$,  all clean models pass the filtering stage with probability at least $0.99$  when the  validation set size has at least $60$ samples.

\begin{figure}[h!]
\centerline{
	\includegraphics[width=0.3\textwidth]{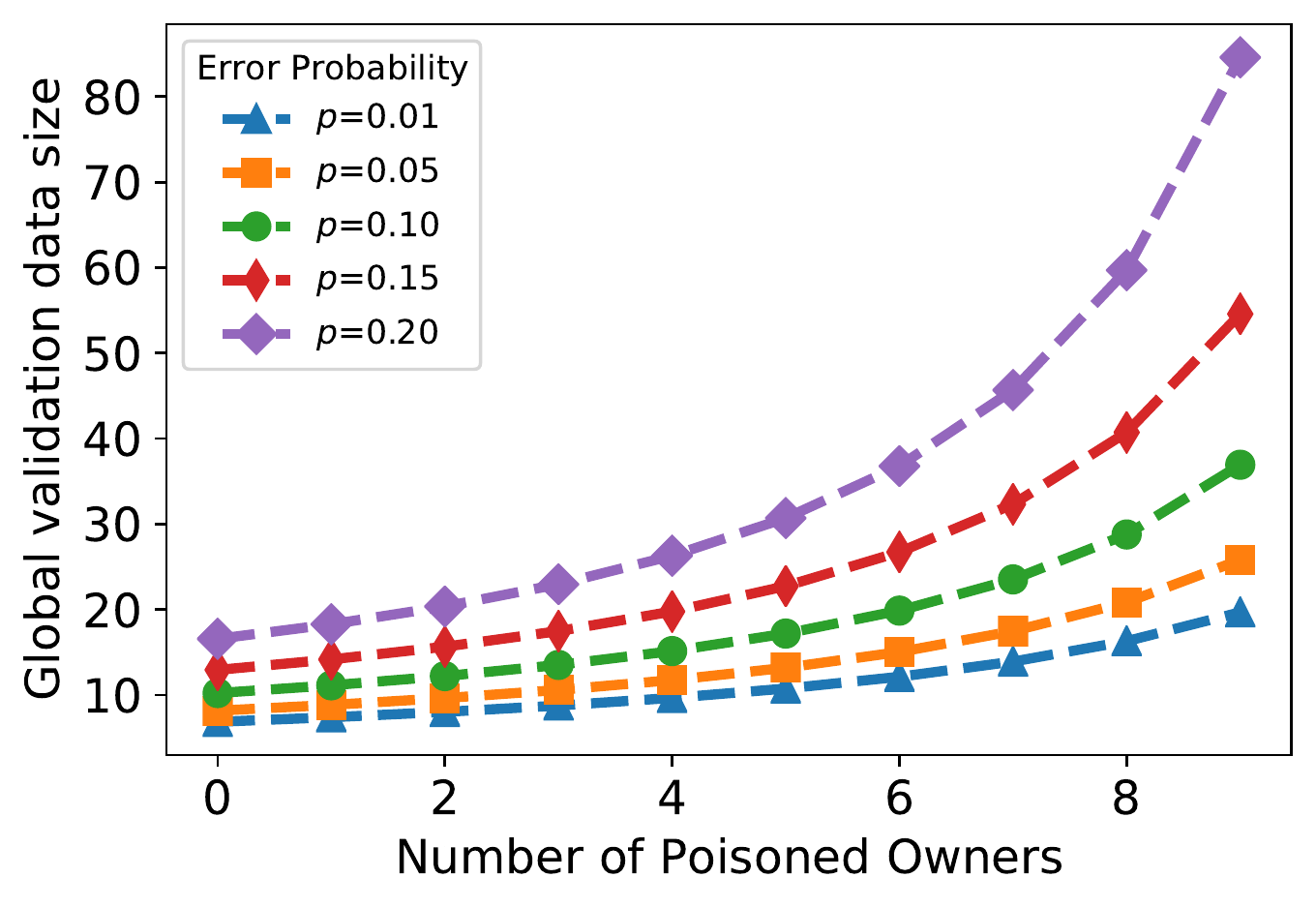}%
}
\caption{ \small Minimum number of samples in the validation dataset as a function of maximum error probability $p$ and number of poisoned owners  $t$ for $m = 20$ data owners. We set the filtering threshold $\threshold=0.03$ and  failure probability $\epsilon=0.01$.}
\label{fig:cv_plots}
\vspace{-2mm}
\end{figure}



 We use a similar strategy as above to compute the lower bound on the test accuracy. On a high level, the proof follows by viewing the combined errors made by the clean models as a Binomial distribution $Bin(m-t, p)$. We can then upper bound the total errors made by all the models in the ensemble by applying Chernoff bounds and consequentially lower bound the ensemble accuracy.

\begin{restatable}{theorem}{LB} \label{thm:LB}
Assume that the conditions in Lemma~\ref{lem:VS} hold against  adversary $\Apsc$ poisoning at most $ t < \frac{m}{2}\frac{1-2p}{1-p}$ owners and that the errors made by the clean models are independent.
	Then $E$  correctly classifies new samples with probability at least $ p_c = (1- \epsilon) \left(1 - e^{-\frac{\delta'^2\mu'}{2+\delta'}}\right)$, where $\mu' = (m-t)p$ and $\delta' = \frac{m-2t}{2\mu'} - 1$.
\end{restatable}

\begin{proof}  \label{proof:LB}
	Lemma~\ref{lem:VS} shows that, with probability $>1-\epsilon$, no clean models will be filtered during ensemble filtering. Given all clean models pass the  filtering stage, we consider the worst case where even the $t$ poisoned models bypass filtering.
	Now, given a new test sample, $m-t$ clean models have uncorrelated errors each with probability at most $p$, the error made by each clean model can be viewed as a Bernoulli random variable with probability $p$ and so the total errors made by clean models follow a binomial $X\sim \text{Bin}(m-t, p)$. We assume that a new sample will be misclassified by all $t$ of the poisoned models. Then the ensemble as a whole makes an error if $t+Bin(m-t, p)>m/2$. We can then bound the probability this occurs by applying Chernoff bound as follows:
\begin{equation*}
\Prob\left[X + t \geq \frac{m}{2}\right] = \Prob\left[X \geq (1 + \delta')\mu'\right] \leq e^{-\frac{\delta'^2\mu'}{2+\delta'}},
\end{equation*}
where $\mu' = (m-t)p$ is the mean of $X$ and $\delta' = \frac{m-2t}{2\mu'}-1 > 0$. Then the probability of making a correct prediction can be lower bounded by:
\begin{equation*}
\Prob\left[X < \frac{m}{2} - t\right] > 1 - e^{-\frac{\delta'^2\mu'}{2+\delta'}},
\end{equation*}
given the number of poisoned models $$t < \frac{m(1-2p)}{2(1-p)}.$$

The inequality on $t$ comes from the constraint $\delta'>0$ for the Chernoff bound to hold. Note that, the above bound holds only when all the clean models pass the filtering stage, which occurs with probability at least $1 - \epsilon$ by Lemma~\ref{lem:VS}. Then the bound on the probability of making a correct prediction by the ensemble can be written as:
\begin{equation*}
	\Prob\left[X < \frac{m}{2} - t\right] > (1- \epsilon) \left(1 - e^{-\frac{\delta'^2\mu'}{2+\delta'}}\right)
\end{equation*}
\end{proof}

\section{Realization in MPC}
\label{sec:mpc_inst}


To instantiate SafeNet in MPC, we first describe the required MPC building blocks, and then provide the SafeNet training and secure prediction protocols.


\smallskip
\subsubsection{MPC Building Blocks} The notation $\share{x}$ denotes a given value $x$ secret-shared among the servers. The exact structure of secret sharing is dependent on the particular instantiation of the underlying MPC framework\cite{DSZ15, AFLNO16, GordonR018, CCPS19, CRS20,BCPS20}. We assume each value  and its  respective secret shares to be elements over an arithmetic ring $\Z{\ell}$. All  multiplication and addition operations are carried out over $\Z{\ell}$.

We express each of our building blocks in the form of an ideal functionality and its corresponding protocol.  An ideal functionality can be viewed as an oracle, which takes input from the parties, applies a predefined function $f$ on the inputs and returns the output back to the parties. The inputs and outputs can be in clear or in $\share{\cdot}$-shared format depending on the definition of the functionality. These ideal functionalities are realized using secure protocols depending on the specific instantiation of the MPC framework agreed upon by the parties. Below are the required building blocks:

\smallskip
\myparagraph{Secure Input Sharing} Ideal Functionality $\FSh$ takes as input  a value $x$ from a party who wants to generate a $\share{\cdot}$-sharing of x, while other parties input $\bot$ to the functionality. $\FSh$ generates a $\share{\cdot}$-sharing of $x$ and sends the appropriate shares to the parties. 
We use $\piSh$ to denote the protocol that realizes this functionality securely.

\smallskip
\myparagraph{Secure Addition} Given $\share{\cdot}$-shares of $x$ and $y$, secure addition is realized by parties locally adding their shares $\share{z} = \share{x} +\share{y}$, where $z = x+y$.


\smallskip
\myparagraph{Secure Multiplication:} Functionality $\FMul$ takes as input $\share{\cdot}$-shares of values $x$ and $y$, creates $\share{\cdot}$-shares of $z = xy$ and sends the  shares of $z$ to the parties. $\piMultA$ denotes the protocol to securely realize $\FMul$.

\smallskip
\myparagraph{Secure Output Reconstruction}
$\FOp$ functionality takes as input  $\share{\cdot}$-shares of a value $x$ from the parties and a commonly agreed upon party id $\textsf{pid}$ in clear. On receiving the shares and $\textsf{pid}$,  $\FOp$ reconstructs $x$ and sends it to the party associated to $\textsf{pid}$. 

\smallskip
\myparagraph{Secure Comparison}
$\FComp$ functionality takes as input a value $a$ in $\share{\cdot}$-shared format. $\FComp$ initializes a bit $b=0$, sets $b = 1$ if $a>0$ and outputs it in $\share{\cdot}$-shared format. Protocol $\piComp$ is used to securely realize $\FComp$.

\smallskip
\myparagraph{Secure Zero-Vector}
$\FZVec$ functionality takes as input a value $L$ in clear from the parties. $\FZVec$  constructs  a vector $\vz$ of all zeros of size $L$ and outputs  $\share{\cdot}$-shares of $\vz$. $\piZVec$ denotes the protocol that securely realizes $\FZVec$.

\smallskip
\myparagraph{Secure Argmax}
$\FAmax$ functionality takes as input a vector $\vx$ in $\share{\cdot}$-shared format and outputs $\share{\cdot}$-shares of a value $\textsc{op}$, where $\textsc{op}$ denotes the index of the max element in vector $\vx$. $\piAMX$ denotes the protocol that securely realizes $\FAmax$.

\smallskip
\subsubsection{ML Building Blocks}


We introduce several building blocks required for private ML training, implemented by
  existing MPC frameworks~\cite{MohasselZ17, MR18, BCPS20, WTBKMR21}:
  
 

\smallskip
\myparagraph{Secure Model Prediction}
$\FmPred$ functionality takes as input a trained model $\fm$ and a feature vector $\vx$ in $\share{\cdot}$-shared format. $\FmPred$ then computes prediction $\textsc{\bf Preds} = \fm(\vx)$ in one-hot vector format and outputs $\share{\cdot}$-shares of the same. $\pimPred$ denotes the protocol which securely realizes functionality $\FmPred$.

\smallskip
\myparagraph{Secure Accuracy}
$\FAcc$ functionality takes as input two equal length vectors
$\vy_{pred}$ and $ \lbl$ in $\share{\cdot}$-shared format. $\FAcc$ then computes the total number matches (element-wise) between the two vectors and outputs $\frac{\#~\text{matches}}{|\lbl|}$ in $\share{\cdot}$-shared format. $\piAcc$ denotes the protocol which securely realizes this functionality.

\smallskip
\subsubsection{Protocols} \label{sec: MLProt}
 We propose two protocols to realize our SafeNet framework in the SOC setting. The first protocol $\pistr$ describes the SafeNet training phase where given $\share{\cdot}$-shares of dataset $\cv{k}$ and model $\fm_k$, with respect to each owner $\usr_k$, $\pistr$  outputs  $\share{\cdot}$-shares of an ensemble $E$ of $m$ models and vector $\valvec$. 
 The second protocol $\piPred$ describes the prediction phase of SafeNet, which given  $\share{\cdot}$-shares of a client's query  predicts its output label. The detailed description for each protocol is as follows:  

\smallskip
\myparagraph{SafeNet Training} \label{sec:sectrain}
 We follow the notation from Algorithm 1.   
  Our goal is for training protocol $\pistr$ given in Figure~\ref{prot:filter} to securely realize functionality $\FpTrain$ (Figure 2), where the inputs to $\FpTrain$ are $\share{\cdot}$-shares of $\ld_k = \cv{k}$ and $a_k = \fm_k$, and the corresponding outputs are $\share{\cdot}$-shares of $O = E$ and $\valvec$. 
  Given the inputs to $\pistr$, the servers first construct a common validation dataset $\share{\CV} = \cup_{k=1}^{m} \share{\cv{k}}$ and an ensemble of models $ \share{E} = \lbrace \share{\fm_k} \rbrace_{k=1}^{m}$.
Then  for each model $\fm_k \in E$,  the servers  compute the validation accuracy $\share{\valacc_k}$. The output $\share{\valacc_k}$ is  compared with a pre-agreed threshold $\threshold$ to obtain a $\share{\cdot}$-sharing of $\valvec_k$, where  $\valvec_k = 1$ if $\valacc_k > \threshold$. After execution of  $\pistr$ protocol,  servers obtain $\share{\cdot}$-shares of ensemble $E$ and vector $\valvec$.

\begin{figure}[h!]
\begin{protocolbox}{$\pistr$}{SafeNet Training Protocol}{prot:filter} 
	\algoHead{Input:} $\share{\cdot}$-shares of  each owner $\usr_k$'s  validation dataset $\cv{k}$ and local model $\fm_k$.
	
	\smallskip
	\algoHead{Protocol Steps:} The servers perform the following:
	\begin{itemize}
		\item[--] Construct $\share{\cdot}$-shares of ensemble $ E = \lbrace \fm_k \rbrace_{k=1}^{m}$ and validation dataset  $\CV = \cup_{k=1}^{m} \cv{k}$.
		
		 \item[--] Execute $\piZVec$ with $m$ as the input and obtain $\share{\cdot}$-shares of a vector $\valvec$.	
		
		\smallskip
		\item[--] For $k \in [1,m]:$
	\begin{itemize}
		
			\item[--] Execute $\pimPred$ with inputs as $\share{\fm_k}$ and $\share{\CV}$ and obtain $\share{\textsc{PREDS}_k}$, where $\textsc{Preds}_k = \fm_k(\CV)$ 
			
			\item[--] Execute $\piAcc$ with inputs as $\share{\textsc{Preds}_k}$ and $\share{\lbl_{\sss{\CV}}}$ and obtain  $\share{\valacc_k}$ as the output.
			
			\item[--] Locally subtract $\share{\cdot}$-shares of $\valacc_k$ with  $\threshold$ to obtain $\share{\valacc_k - \threshold }$.
			
			\item[--] Execute $\piComp$ with input as $\share{\valacc_k - \threshold}$ and obtain $\share{b'}$, where $b' = 1$ iff $ \valacc_k > \threshold$. Set the $k^{\text{th}}$ position in $\share{\valvec}$ as $\share{\valvec_k} = \share{b'}$ 
			
 			
			
			
	\end{itemize}
	

	\end{itemize}  
	
	\smallskip
	\algoHead{Output:} $\share{\cdot}$-shares of $\valvec$ and ensemble $E$.
	
\end{protocolbox}
\end{figure}

The security proof of  $\pistr$ protocol  as  stated in Theorem~\ref{thm:SNtrain} in Section~\ref{sec:train} is given  in Appendix~\ref{proof:SPfilter}.

\begin{figure}[h!]
	\begin{protocolbox}{$\piPred$}{SafeNet Prediction Protocol}{prot:pred} 
		\algoHead{Input:} $\share{\cdot}$-shares of vector $\valvec$ and ensemble $E$ among the servers. Client $\share{\cdot}$-shares query $\vx$ to the servers. 
		
		\smallskip
		\algoHead{Protocol Steps:} The servers perform the following:
		\begin{itemize}
			
			\item[--] Execute $\piZVec$ protocol with $L$ as the input, where $L$ denotes the number of distinct class labels and obtain $\share{\cdot}$-shares of $\vz$.	
			
			
			\item[--] For each $\fm_k \in E:$
			\begin{itemize}
				
				\item[--] Execute $\pimPred$ with inputs as $\share{\fm_k}$ and $\share{\vx}$. Obtain $\share{\textsc{\bf Preds}}$, where $\textsc{\bf Preds} = \fm_k(\vx)$. 
				
				\item[--] Execute $\piMultA$ to multiply $\valvec_k$ to each element of vector ${\textsc{\bf Preds}}$.
				
				\item[--] Locally add $\share{\vz} = \share{\vz} +  \share{\textsc{\bf Preds}}$ to update $\vz$. 
				
			\end{itemize}
			
			\item[--] Execute $\piAMX$ protocol with input as $\share{\vz}$ and obtain $\share{\textsc{op}}$ as the output.

			
			%
		\end{itemize}  
		
		\smallskip
		\algoHead{Output:}  $\share{\cdot}$-shares of $\textsc{op}$
		
	\end{protocolbox}
\end{figure}

\smallskip
\myparagraph{SafeNet Prediction}  Functionality $\FPred$  takes as input party id $\textsf{cid}$, $\share{\cdot}$-shares of client query $\vx$, vector $\valvec$ and  ensemble $E = \lbrace \share{\fm_k} \rbrace_{k=1}^{m}$ and outputs a value  $\textsc{op}$, the predicted class label by ensemble $E$ on query $\vx$.

 Protocol $\piPred$ realizes $\FPred$ as follows:
 Given  $\share{\cdot}$-shares of $\vx$,  $\valvec$ and ensemble $E$,
the  servers initialize a vector $\vz$ of all zeros of size $L$. For each model $\fm_k$ in the ensemble $E$, the servers compute $\share{\cdot}$-shares of  the prediction ${\textsc{\bf Preds}} =  \fm_k(\vx)$ in one-hot format. The element $\valvec_k$ in vector $\valvec$ is multiplied to each element in vector ${\textsc{\bf Preds}}$. The $\share{{\textsc{\bf Preds}}}$ vector is added to $\share{\vz}$ to update the model's vote towards the final prediction. If $\valvec_k = 0$, then after multiplication  vector ${\textsc{\bf Preds}}$  is a vector of zeros and does not contribute in the  voting process towards the final prediction. 
 The servers then compute the argmax of vector $\share{\vz}$ and receive  output $\share{\textsc{op}}$ from $\piAMX$, where  $\textsc{op}$ denotes the predicted class label by the ensemble. The appropriate $\share{\cdot}$-shares of $\textsc{op}$ is forwarded to the client for reconstruction.



\begin{restatable}{theorem}{SPpred} \label{thm:SPpred}
	Protocol $\piPred$ is secure  against adversary $\Apsc$ who poisons  $t$ out of $m$ data owners and corrupts $T$ out of $N$ servers.
\end{restatable}

\begin{proof}
	The proof is given below in Appendix~\ref{proof:SPpred}.
\end{proof}	

\section{Security Proofs} 
\label{app:secproof}

For concise security proofs, we assume the adversary $\Apsc$ performs a semi-honest corruption in the SOC paradigm, but our proofs can also be extended to malicious adversaries in the MPC. We prove that  protocol $\pistr$  is secure against an adversary of type $\Apsc$. Towards this, we first argue that protocol $\pistr$ securely realizes the standard ideal-world functionality $\FpTrain$. We use simulation based security to prove our claim.
 Next, we argue that the ensemble $E$ trained using $\pistr$ protocol provides poisoning robustness against $\Apsc$.


\smallskip
\SPtrain*

\begin{proof} \label{proof:SPfilter}
	Let $\Apsc$ be a real-world adversary that semi-honestly corrupts $T$ out of $N$ servers at the beginning of the protocol $\pistr$. We now present the steps of the ideal-world adversary (simulator) $\Sim_f$ for $\Apsc$. Note that, in the semi-honest setting $\Sim_f$ already posses the input of $\Apsc$ and the final output shares of  $\valvec$. $\Sim_f$ acts on behalf of $N-T$ honest servers, sets their shares as random values in $\Z{\ell}$ and simulates each step of $\pistr$ protocol to the corrupt servers as follows:

	\begin{itemize}
		
		\item[--] No simulation is required to construct $\share{\cdot}$-shares of ensemble $E$ and validation dataset $\CV$ as it happens locally.
		
		\smallskip
		\item[--] $\Sim_f$ simulates messages on behalf of honest servers as a part of the protocol steps of $\piZVec$ with public value $m$ as the input and eventually sends and receives appropriate $\share{\cdot}$-shares of $\valvec$ to and from $\Apsc$.
		
		\smallskip
		\item[--] For $k \in [1,m]$:
		\begin{itemize}
			\smallskip
			\item[--] $\Sim_f$ simulates messages on behalf  of honest servers, as a part of the protocol steps of $\pimPred$, with inputs to the protocol as $\share{\cdot}$-shares of $\fm_k$ and $\CV$ and eventually sends and receives appropriate $\share{\cdot}$-shares of $\textsc{PREDS}_k$ to and from $\Apsc$.
			
			\smallskip
			\item[--] $\Sim_f$ simulates messages on behalf  of honest servers, as a part of the protocol steps of $\piAcc$, with inputs to the protocol as  $\share{\cdot}$-shares of $\textsc{PREDS}_k$ and $\lbl_{\CV}$ and eventually sends and receives appropriate $\share{\cdot}$-shares of $\valacc_k$ to and from $\Apsc$.
			
			\smallskip
			\item[--] No simulation is required for subtraction with threshold $\threshold$ as it happens locally.
			
			\smallskip
			\item[--] $\Sim_f$ simulates messages on behalf  of honest servers, as a part of the protocol steps of $\piComp$, with inputs to the protocols as $\share{\cdot}$-shares of $\valacc-\threshold$ and at the end $\Sim_f$ instead sends the original shares of $\valvec_k$  as shares of $b'$ associated to $\Apsc$.
			
			\smallskip
			\item[--] No simulation is required to assign $\share{\valvec_k} = \share{b'}$.
		
		\end{itemize}
		
		
		
	\end{itemize}
The proof now simply follows from the fact that  simulated view and real-world view of the adversary are computationally indistinguishable and concludes that $\pistr$ securely realizes functionality $\FpTrain$.

Now given the output of $\pistr$ protocol is an ensemble $E$, we showed in the proof of Theorem \ref{thm:LB} that $E$ correctly classifies a sample with probability at least $p_c$. As a result the underlying trained model also provides  poisoning robustness  against $\Apsc$.
	 
\end{proof}

We use a similar argument to show protocol $\piPred$ is secure against adversary $\Apsc$. 
\smallskip
\SPpred*
\begin{proof} \label{proof:SPpred}
Let $\Apsc$ be a real-world adversary that poisons $t$ out of $m$ owners and semi honestly corrupts $T$ out of $N$ servers at the beginning of $\piPred$ protocol. We present steps of the ideal-world adversary (simulator) $\Sim_f$ for $\Apsc$. $\Sim_f$ on behalf of the honest servers, sets their shares as random values in $\Z{\ell}$ and simulates each step of $\piPred$ protocol to the corrupt servers as follows:
\begin{itemize}
	\item[--]  $\Sim_f$ simulates messages on behalf of honest servers as a part of the protocol steps of $\piZVec$ with public value $L$ as the input and eventually sends and receives appropriate $\share{\cdot}$-shares of $\vz$ to and from $\Apsc$.

	\smallskip
	\item[--] For $k \in [1,m']$:
	
	\begin{itemize}
		\smallskip
		\item[--] $\Sim_f$ simulates messages on behalf  of honest servers, as a part of the protocol steps of $\pimPred$, which takes input as $\share{\cdot}$-shares of $\fm_k$ and $\vx$. $\Sim_f$ eventually sends and receives appropriate $\share{\cdot}$-shares of ${\bf Preds}$ to and from $\Apsc$.
		
		\smallskip
		\item[--] For every multiplication of $\share{\valvec_k}$ with respect to  each element in ${\bf Preds}$,  $\Sim_f$ simulates messages on behalf  of honest servers, as a part of the protocol steps of $\piMultA$, which takes input as $\share{\cdot}$-shares of ${\bf Preds }_j$  and $\valvec_k$. $\Sim_f$ eventually sends and receives appropriate $\share{\cdot}$-shares of $\valvec_k  \times {\bf Preds}_j$ to and from $\Apsc$.
		
		\smallskip
		\item[--] No simulation is required to update $\share{\vz}$ as addition happens locally.
		
	\end{itemize} 
	
	\item[--] $\Sim_f$ simulates messages on behalf  of honest servers, as a part of the protocol steps of $\piAMX$, which takes input as $\share{\cdot}$-shares of $\vz$. At the end $\Sim_f$ instead forwards the original $\share{\cdot}$-shares of $\textsc{op}$ associated to  $\Apsc$.
	
	
	
\end{itemize}

\noindent
The proof now simply follows from the fact that simulated view and real-world view of the adversary are computationally indistinguishable. Poisoning robustness argument  follows from the fact that the ensemble $E$ used for prediction was trained using protocol $\pistr$ which was shown to be secure against $\Apsc$ in Theorem \ref{thm:SNtrain}.
	\end{proof}	

This concludes the security proofs of our training and prediction protocols.
\section{SafeNet Extensions} \label{apndx:addalgs}

\subsection{Inference phase in Transfer Learning Setting} \label{apndx:TransferLearning}
We provide a modified version of  SafeNet's Inference algorithm in the transfer learning setting, to improve the running time and communication complexity of SafeNet.  Algorithm 3 provides the details of SafeNet's prediction phase below. 

\begin{algorithm}[h!]
	
	\caption*{{\bf Algorithm 3} SafeNet Inference for Transfer Learning Setting} 
	\begin{algorithmic}
		
		\State Input:  Secret-shares of backbone model $\fm_B$,  ensemble of $m$ fine-tuned models $E = \{\fm_1, \ldots, \fm_m\}$, vector $\valvec$ and client query $\vx$.
		
		\State{\Comment{\footnotesize MPC computation in secret-shared format}}
		
		\item[--] Construct vector $\vz$ of all zeros of size $L$, where $L$ denotes the number of distinct class labels. 
		
		\State -- Run forward pass on $\fm_B$ with input  $\vx$ upto its last $l$ layers, where $\vp$ denotes the output vector from that layer.
		
		\State -- For $k \in [1,m]:$
		
		\begin{itemize}
			\item[-] Run forward pass on the last $l$ layers of $\fm_k$ with input as $\vp$. Let the output of the computation be $\textsc{\bf Preds}$, which is one-hot encoding of the predicted label.
			\item[-] Multiply $\valvec_k$ to each element of $\textsc{\bf Preds}$. 
			\item[-] Add $\vz = \vz + \textsc{\bf Preds}$.
		\end{itemize}
	
		\State -- Run argmax  with input as $\vz$ and obtain $\textsc{op}$ as the output.

		\State \Return $\textsc{op}$
	\end{algorithmic}
	\label{alg:tl_pred}
\end{algorithm} 

\subsection{Training with Computationally Restricted Owners} \label{apndx:TrainCRO}
 In this section we provide a modified version of  SafeNet's Training Algorithm, to accommodate when a subset of data owners are computationally restricted, i.e., they can not train their models locally.  Algorithm 4 provides the details of SafeNet's training steps below.

\begin{algorithm}[h!]
	
	\caption*{{\bf Algorithm 4} SafeNet Training with Computationally Restricted Owners} 
	\begin{algorithmic}
		
		\State Input:  $m$ total data owners of which $m_{r}$ subset of owners are computationally restricted, each owner $\usr_k$'s   dataset $\ld_k$. 	
		
		\State{\Comment{\footnotesize Computationally Restricted Owner's local computation in plaintext } }
		\State -- For $k \in [1, m_{r}]:$	
		\begin{itemize}
			\item[-] Separate out $\cv{k}$ from $\ld_k$.
			\item[-] Secret-share cross-validation dataset $\cv{k}$ and training dataset $\ld_k \setminus \cv{k}$ to servers.
		\end{itemize}

		\State{\Comment{\footnotesize Computationally Unrestricted Owner's local computation in plaintext } }
		\State -- For $k \in [m_{r+1},m]:$	
		\begin{itemize}
			\item[-] Separate out $\cv{k}$ from $\ld_k$.  Train $\fm_k$ on $\ld_k \setminus \cv{k}$.
			\item[-] Secret-share $\cv{k}$ and $\fm_k$ to servers.
		\end{itemize}
		
		
		\State{\Comment{\footnotesize MPC computation in secret-shared format}}
		
		\State  1. For $k \in [1,m_r]:$ 
		\begin{itemize}
			\item[-] Train $\fm_k$ on $\ld_k \setminus \cv{k}$.
		\end{itemize}
		
		\State 2. Construct a common validation dataset $\CV = \cup_{i=1}^{m} \cv{i}$ and collect ensemble of models $ E = \lbrace \fm_i \rbrace_{i=1}^{m}$
		
		\State 3. Initialize a vector $\valvec$ of zeros and of size $m$. 
		
		\State  4. For $k \in [1,m]:$ 
		\begin{itemize}
			\item[-] $\valacc_k = Accuracy(\fm_k, \CV)$ 
			
			
			
			\item[-] If $\valacc_k> \threshold$: 
			\begin{itemize}
				\item[--] Set $\valvec_k = 1$   
				
			\end{itemize}			
			
		\end{itemize}
		\State \Return $E$ and $\valvec$
	\end{algorithmic}
	\label{alg:crtrain}
\end{algorithm} 

\section{Additional Experiments}
\label{app:bench}

\subsection{Evaluation of SafeNet Extensions} \label{sec:apndx_extension}

\paragraph{Integration Testing}

Here, we  evaluate the performance of SafeNet by varying the concentration parameter $\alpha$ to manipulate the degree of data similarity among the owners. The experiments are performed with the same neural network architecture from Section~\ref{sec:DNN} on the Fashion dataset. Figure~\ref{fig:Dirchlet} gives a comprehensive view of the variation in test accuracy and attack success rate for backdoor and targeted attacks over several values of $\alpha$. 

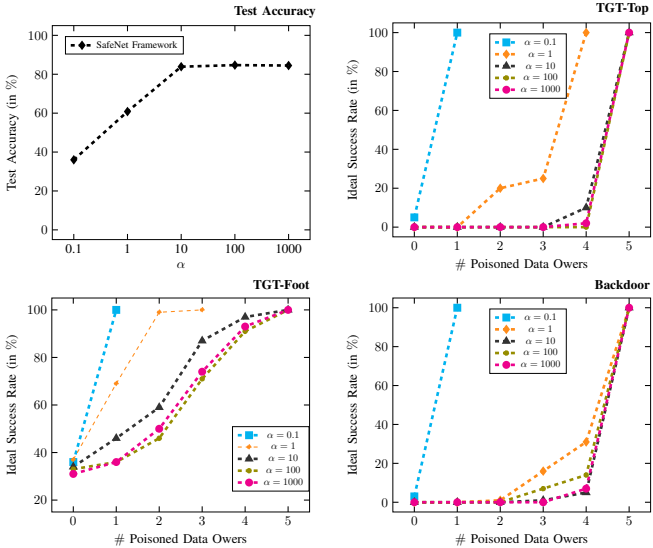
\begin{figure}[thb!]
	\begin{minipage}{.25\textwidth}
		\scalebox{0.5}{
			\begin{tikzpicture}
				\begin{axis}[legend style={at={(0.05,0.95)},anchor=north west , nodes={scale=0.75, transform shape}}, xlabel={$\alpha$}, ylabel={Test Accuracy (in $\%$)}, ymin = -5, ymax = 105, xtick = data, symbolic x coords={0.01, 0.1, 1, 10, 100, 1000}]
					
					\addplot[color = black,every mark/.append style={solid}, mark=diamond,dashed,line width = 2] coordinates { 
					(0.1,36.01) (1,60.82) (10, 83.88) (100,84.70) (1000, 84.48)};
					\addlegendentry{SafeNet Framework}

				\end{axis}
			\node[align=right,font=\bfseries, xshift=9.0em, yshift=1.0em] (title) at (current bounding box.north) {Test Accuracy};
			\end{tikzpicture}
		}
	\end{minipage}%
	\begin{minipage}{.25\textwidth}
		\scalebox{0.5}{
			\begin{tikzpicture}
				\begin{axis}[legend style={at={(0.375,0.95)},anchor=north west , nodes={scale=0.75, transform shape}}, xlabel={$\#$  Poisoned Data Owers}, ylabel={Ideal Success Rate (in $\%$)}, ymin = -5, ymax = 105, xtick = {0,1,2,3,4,5}]
					
					
					\addplot[color = cyan!90, every mark/.append style={solid}, mark=square*,dashed,line width = 2] plot coordinates {(0,5) (1,100)};
					\addlegendentry{$\alpha = 0.1$}
					
					\addplot[color = orange!90,mark = diamond*,every mark/.append style={solid},dashed,line width = 2] plot coordinates {(0,0) (1,0) (2,20) (3,25) (4,100)};
					\addlegendentry{$\alpha=1$}
					
					\addplot[color = black!80,mark=triangle*,every mark/.append style={solid},dashed,line width = 2] plot coordinates {(0,0) (1,0) (2, 0) (3,0) (4,10) (5,100)};
					\addlegendentry{$\alpha= 10$}
					
					\addplot[color = olive, mark=star,dashed,every mark/.append style={solid}, line width = 2] plot coordinates {(0,0) (1,0) (2, 0) (3,0) (4,0) (5,100)};
					\addlegendentry{$\alpha= 100$}
					
					\addplot[color = magenta, every mark/.append style={solid}, mark=*,dashed, every mark/.append style={solid}, line width = 2] plot coordinates {(0,0) (1,0) (2, 0) (3,0) (4,2) (5,100)};
					\addlegendentry{$\alpha= 1000$}

				\end{axis}
				\node[align=right,font=\bfseries, xshift=9.5em, yshift=1.0em] (title) at (current bounding box.north) {TGT-Top};
			\end{tikzpicture}
		}
	\end{minipage}%

	\begin{minipage}{.25\textwidth}
	\scalebox{0.5}{
		\begin{tikzpicture}
			\begin{axis}[legend style={at={(0.7,0.4)},anchor=north west , nodes={scale=0.75, transform shape}}, xlabel={$\#$  Poisoned Data Owers}, ylabel={Ideal Success Rate (in $\%$)}, ymin = 15, ymax = 105, xtick = {0,1,2,3,4,5}]
				

				\addplot[color = cyan!90, mark=square*,dashed, every mark/.append style={solid}, line width = 2] plot coordinates {(0,36) (1,100)};
				\addlegendentry{$\alpha = 0.1$}
				
				\addplot[color = orange!90,mark=diamond*,dashed, thick] plot coordinates {(0,37) (1,69) (2,99) (3,100)};
				\addlegendentry{$\alpha = 1$}
				
				\addplot[color = black!80,mark=triangle*,dashed,every mark/.append style={solid}, line width = 2] plot coordinates {(0,34) (1,46) (2, 59) (3,87) (4,97) (5,100)};
				\addlegendentry{$\alpha = 10$} 
				
				\addplot[color = olive, mark=star,dashed,every mark/.append style={solid}, line width = 2] plot coordinates {(0,33) (1,36) (2, 46) (3,71) (4,91) (5,100)};
				\addlegendentry{$\alpha= 100$}
				
				\addplot[color = magenta, every mark/.append style={solid}, mark=*,dashed,every mark/.append style={solid}, line width = 2] plot coordinates {(0,31) (1,36) (2, 50) (3,74) (4,93) (5,100)};
				\addlegendentry{$\alpha= 1000$}

			\end{axis}
			\node[align=right,font=\bfseries, xshift=9.5em, yshift=1.0em] (title) at (current bounding box.north) {TGT-Foot};
		\end{tikzpicture}
	}
	\end{minipage}%
	\begin{minipage}{.25\textwidth}
		\scalebox{0.5}{
			\begin{tikzpicture}
			\begin{axis}[legend style={at={(0.375,0.95)},anchor=north west , nodes={scale=0.75, transform shape}}, xlabel={$\#$  Poisoned Data Owers}, ylabel={Ideal Success Rate (in $\%$)}, ymin = -5, ymax = 105, xtick = {0,1,2,3,4,5}]
					
					
					\addplot[color = cyan!90, mark=square*,dashed, every mark/.append style={solid}, line width = 2] plot coordinates {(0,3) (1,100)};
					\addlegendentry{$\alpha = 0.1$}
					
					\addplot[color = orange!90,mark=diamond*,dashed, thick,every mark/.append style={solid}, line width = 2] plot coordinates {(0,0) (1,0) (2,1) (3,16) (4,31) (5,100)};
					\addlegendentry{$\alpha=1$}
					
					\addplot[color = black!80,mark=triangle*,dashed,every mark/.append style={solid}, line width = 2] plot coordinates {(0,0) (1,0) (2, 0) (3,1) (4,5) (5,100)};
					\addlegendentry{$\alpha = 10$}
					
					\addplot[color = olive, mark=star,dashed, every mark/.append style={solid}, line width = 2] plot coordinates {(0,0) (1,0) (2, 0) (3,7) (4,14) (5,100)};
					\addlegendentry{$\alpha= 100$}
					
					\addplot[color = magenta, mark= *,dashed, every mark/.append style={solid}, line width = 2] plot coordinates {(0,0) (1,0) (2, 0) (3,0) (4,7) (5,100)};
					\addlegendentry{$\alpha= 1000$}

				\end{axis}
				\node[align=right,font=\bfseries, xshift=9.5em, yshift=1.0em] (title) at (current bounding box.north) {Backdoor};
			\end{tikzpicture}
		}
	\end{minipage}
	\caption{\small Test Accuracy and Worst-case Adversarial Success in a three layer neural network model trained on Fashion dataset using  SafeNet  for varying data distributions. Parameter $\alpha$ dictates the similarity of distributions between the owners. Higher values of $\alpha$ denote greater similarity in data distributions among the owners and results in increased SafeNet robustness.} \label{fig:Dirchlet}
	\vspace{-2mm}
\end{figure}

We observe that as $\alpha$ decreases, i.e., the underlying data distribution of the owners becomes more non-iid, the test accuracy of SafeNet starts to drop. This is expected as there will be less agreement between the different models, and the majority vote will have a larger chance of errors. In such cases it is easier for the adversary to launch an attack as there is rarely any agreement among the models in the ensemble, and the final output is swayed towards the target label of attackers' choice. Figure~\ref{fig:Dirchlet} shows that for both targeted and backdoor attacks,  SafeNet  holds up  well until  $\alpha$ reaches extremely small values ($\alpha = 0.1$), at which point we observe the robustness break down.
However, the design of  SafeNet  allows us to detect difference in owners' distributions at early stages of our framework. 
For instance, we experiment for $\alpha = 0.1$ and observe that the average $\valacc$ accuracy of the models is $17\%$. Such low accuracies for most of the models in the ensemble indicate non-identical distributions and we recommend not to use SafeNet in such cases.
  

\paragraph{Low Resource Users}
We instantiate our Fashion dataset setup in the 3PC setting and assume $2$ out of $10$ data owners are computationally restricted. We observe SafeNet still runs  $1.82\times$ faster and requires $3.53\times$ less communication compared to the existing PPML framework, while retaining its robustness against poisoning and privacy attacks.

{\begin{table*}[h!]
		
		\centering \scriptsize
		\begin{adjustbox}{max width=\textwidth}{  
				\begin{tabular}{c  c  c  r  r  r  r  r  r r}
					
					
					\multirow{2}{*}{MPC} & \multirow{2}{*}{Setting} & \multirow{2}{*}{Framework} &  \multicolumn{1}{c}{\multirow{2}{*}{Training Time (s)}} & \multicolumn{1}{c}{\multirow{2}{*}{Communication (GB)}} & \multicolumn{2}{c}{Backdoor Attack} & \multicolumn{3}{c}{Targeted Attack}\\ \cmidrule{6-10}
					
					&   &   &   &  & Test Accuracy & Success Rate  & Test Accuracy & Success Rate-Top & Success Rate-Foot\\
					
					\midrule
					
					\multirow{2}{*}{3PC \cite{AFLNO16}} & \multirow{2}{*}{Semi-Honest} & PPML  & n$\times$243.55 & n$\times$55.68 & $89.14\%$ & $100\%$ & $87.34\%$ & $83\%$ & $90\%$\\
					
					& & SafeNet & $10.03$ & $2.05$ & $88.68\%$ & $4\%$ & $88.65\%$ & $1\%$ & $10\%$\\
					

					\midrule
					
					\multirow{2}{*}{4PC \cite{DEK21}} & \multirow{2}{*}{Malicious} & PPML  & n$\times$588.42 & n$\times$105.85 & $89.14\%$ & $100\%$ & $87.22\%$ & $83\%$ & $90\%$\\
					
					&	& SafeNet & $23.39$ & $3.78$ & $88.65\%$ & $4\%$ & $88.65\%$ & $1\%$ & $10\%$\\
					

				\end{tabular}
			}
		\end{adjustbox}
		\caption{\small  Training time (in seconds) and Communication (in GB) over a LAN network for  traditional PPML and SafeNet framework training a  multiclass logistic regression on MNIST. n denotes the number of epochs in the PPML framework. The time and communication reported for SafeNet is for end-to-end execution. Test Accuracy and Success Rate are given for a single poisoned owner.}
		\label{tab:MNIST_logreg}
		
	\end{table*}
	
}

\smallskip
\subsection{Logistic Regresssion, Multiclass Classification} \label{sec:log_multiclass}
We use the same strategies for the Backdoor and Targeted attacks on the MNIST dataset. For  BadNets, we select the initial class $y_s=4$ and the target label $y_t=9$, and use the same $y_t=9$ for the targeted attack.
 Table \ref{tab:MNIST_logreg} provides a detailed analysis of the training time, communication, test accuracy, and success rate for both frameworks, in presence of a single poisoned owner. 
 The worst-case adversarial success for SafeNet  is in Figure~\ref{fig:RA_LogMultClass}. The slow rise in the success rate of the adversary across multiple attacks shows the robust accuracy property of our framework translates smoothly  for the case of a multi-class classification problem.   



\begin{figure}[htb!]
	\centering
	\begin{minipage}{0.3\textwidth}
		\scalebox{0.6}{
			\begin{tikzpicture}
				\begin{axis}[legend style={at={(0.05,0.95)},anchor=north west , nodes={scale=0.75, transform shape}}, xlabel={$\#$ Poisoned Data Owers}, ylabel={Ideal Success Rate (in $\%$)}, ymin = -5, ymax = 105, xtick = data,height=6.5cm,width=9cm]
					
					\addplot[color = cyan!90,mark=triangle*,dashed,every mark/.append style={solid}, line width = 2] coordinates { (0,1) (1,1) (2, 1) (3,1) (4,2) (5,2) (6,2) (7,2) (8,2)(9,4)(10,100)};
					\addlegendentry{SafeNet-TGT-Top}

					\addplot[color = red!90,every mark/.append style={solid}, mark=diamond*,dashed,every mark/.append style={solid}, line width = 2] plot coordinates {(0,11) (1,12) (2,13) (3,13) (4,16) (5,18) (6,21) (7,27) (8,33) (9,48) (10,100)};
					\addlegendentry{SafeNet-TGT-Foot}
					
					\addplot[color = black!75, every mark/.append style={solid}, mark=square*,dashed,every mark/.append style={solid}, line width = 2] coordinates { (0,0) (1,0) (2, 0) (3,0) (4,1) (5,2) (6,4) (7,6) (8,10)(9,17)(10,100)};
					\addlegendentry{SafeNet-Backdoor}
					
					


				\end{axis}
			\end{tikzpicture}
		}
	\end{minipage}%
		\caption{\small Worst-case adversarial success of multi-class logistic regression  on MNIST in the SafeNet framework for backdoor and targeted attacks. The adversary can change the set of $c$ poisoned owners per sample. SafeNet achieves certified robustness up to 9 poisoned owners out of 20 against backdoor and TGT-TOP attacks. The TGT-Foot attack targeting low-confidence samples has slightly higher success, as expected.} \label{fig:RA_LogMultClass}	
	\vspace{-2mm}
\end{figure}
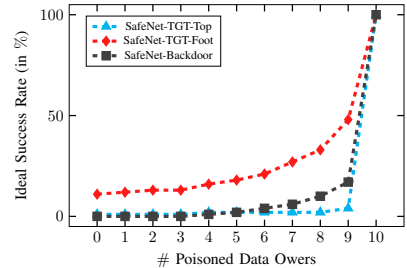

\subsection{Evaluation on Deep Learning Models}
\smallskip
\myparagraph{Experiments on Fashion Dataset} We present  results on  one and two layer deep neural networks trained on the Fashion dataset. 
%
We perform the same set of backdoor and targeted attacks as described in Section~\ref{sec:exp}. Tables~\ref{tab:Fashion_nn12TC} and \ref{tab:Fashion_nn12_testacc} provide detailed analysis of the training time, communication, test accuracy, and success rate for traditional PPML and SafeNet frameworks.  We observe similar improvements, where for instance in the 4PC setting, SafeNet has $42\times$ and $43\times$ improvement in training time and communication complexity over the PPML framework, for $n=10$ epochs for a two hidden layer neural network. Figure~\ref{fig:RA_APDXFASHION} shows the worst-case attack success in SafeNet (where the attacker can choose the subset of corrupted owners per sample) and the results are similar to Figure~\ref{fig:RA_FASHION}.

\begin{figure}[h!]
	\begin{minipage}{.25\textwidth}
		\scalebox{0.5}{
			\begin{tikzpicture}
				\begin{axis}[legend style={at={(0.02,0.98)},anchor=north west , nodes={scale=0.6, transform shape}}, xlabel={$\#$ Poisoned Data Owers}, ylabel={Ideal Success Rate (in \%)}, ymin = -5, ymax = 105, xtick = data]
					
					\addplot[color = red,mark=*,dashed,every mark/.append style={solid}, line width = 2] coordinates { (0,0) (1,0) (2, 0) (3,0) (4,0) (5,100)};
					\addlegendentry{SafeNet-TGT-Top}
					
					\addplot[color = black!80,mark=triangle*,dashed,every mark/.append style={solid}, line width = 2] plot coordinates {(0,12) (1,18) (2, 21) (3,26) (4,44) (5,100)};
					\addlegendentry{SafeNet-TGT-Random}
					
					\addplot[color = orange!90,mark=diamond*,dashed,every mark/.append style={solid}, line width = 2] plot coordinates {(0,36) (1,54) (2,66) (3,87) (4,93) (5,100)};
					\addlegendentry{SafeNet-TGT-Foot}
					
					\addplot[color = cyan!90,mark=square*,dashed,every mark/.append style={solid}, line width = 2] plot coordinates {(0,0) (1,0) (2,0) (3,2) (4,15) (5,100)};
					\addlegendentry{SafeNet-Backdoor}


				\end{axis}
				\node[align=center,font=\bfseries, xshift=-3.5em, yshift=1.0em] (title) at (current bounding box.north) {1-Layer NN};
			\end{tikzpicture}
		}
	\end{minipage}%
	\begin{minipage}{.25\textwidth}
		\scalebox{0.5}{
			\begin{tikzpicture}
				\begin{axis}[legend style={at={(0.02,0.98)},anchor=north west , nodes={scale=0.6, transform shape}}, xlabel={$\#$ Poisoned Data Owers}, ylabel={Ideal Success Rate (in $\%$)}, ymin = -5, ymax = 105, xtick = data]
					
					\addplot[color = red,mark=*,dashed,every mark/.append style={solid}, line width = 2] coordinates { (0,0) (1,0) (2, 0) (3,0) (4,0) (5,100)};
					\addlegendentry{SafeNet-TGT-Top}
					
					\addplot[color = black!80,mark=triangle*,dashed,every mark/.append style={solid}, line width = 2] plot coordinates {(0,10) (1,17) (2, 21) (3,28) (4,36) (5,100)};
					\addlegendentry{SafeNet-TGT-Random}
					
					\addplot[color = orange!90,mark=diamond*,dashed,every mark/.append style={solid}, line width = 2] plot coordinates {(0,41) (1,52) (2,72) (3,88) (4,99) (5,100)};
					\addlegendentry{SafeNet-TGT-Foot}
					
					\addplot[color = cyan!90,mark=square*,dashed,every mark/.append style={solid}, line width = 2] plot coordinates {(0,0) (1,0) (2,1) (3,11) (4,28) (5,100)};
					\addlegendentry{SafeNet-Backdoor}


				\end{axis}
				\node[align=center,font=\bfseries, xshift=-3.5em, yshift=1.0em] (title) at (current bounding box.north) {2-Layer NN};
			\end{tikzpicture}
		}
	\end{minipage}%
	\caption{\small Worst-case adversarial success  of one and two layer Neural Networks  on FASHION  dataset in  SafeNet  framework for varying poisoned owners.} \label{fig:RA_APDXFASHION}

\end{figure}
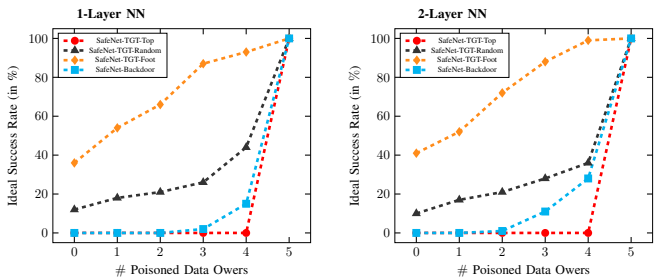

{\begin{table}[h!]
		\begin{adjustbox}{max width= 0.48\textwidth}{  
				\begin{tabular}{c  c  c	 c  r  r}
					
					
					\multirow{1}{*}{MPC } & \multirow{1}{*}{ Setting} & \multirow{1}{*}{No. Hidden Layers} & \multirow{1}{*}{Framework} &  \multicolumn{1}{c}{\multirow{1}{*}{Training Time (s)}} & \multicolumn{1}{c}{\multirow{1}{*}{Communication (GB)}} \\ \cmidrule{1-6}

					\multirow{5}{*}{3PC \cite{AFLNO16}} & \multirow{5}{*}{Semi-Honest} & 	\multirow{2}{*}{1} & PPML  & n$\times$382.34 & n$\times$ 96.37 \\
					
					&	& & SafeNet & $65.71$ & $14.58$\\
					
					\cmidrule{3-6}
					
					&	& \multirow{2}{*}{2} & PPML & n$\times$474.66 & n$\times$ 125.58\\
					
					&	&  & SafeNet & $108.12$ & $27.98$\\
					
					
					
					
					\midrule
					
					\multirow{5}{*}{4PC \cite{DEK21}} & \multirow{5}{*}{Malicious} & 	\multirow{2}{*}{1} & PPML  & n$\times$869.12 & n$\times$ 174.12\\
					
					&	& & SafeNet & $152.68$ & $26.89$\\
					
					\cmidrule{3-6}
					
					&	& \multirow{2}{*}{2} & PPML & n$\times$1099.06 & 
					n$\times$227.23\\
					
					&	&  & SafeNet & $258.72$ & $51.66$\\
					
					
					

				\end{tabular}
			}
		\end{adjustbox}
        
		\caption{\small  Training Time (in seconds) and Communication (in GB) of PPML and SafeNet frameworks for  one and two  layer neural network  on  Fashion dataset, where n denotes the number of epochs.  The time and communication reported for SafeNet framework is for end-to-end execution. }
		\label{tab:Fashion_nn12TC}
		
	\end{table}
	
}

{\begin{table}[h!]
		\begin{adjustbox}{max width= 0.49 \textwidth}{  
				\begin{tabular}{c  c  c	 c  r  r  r  r}
					
					
					\multirow{2}{*}{MPC} & \multirow{2}{*}{ Setting} & \multirow{2}{*}{No. Hidden Layers} & \multirow{2}{*}{Framework} &\multirow{2}{*}{Test Accuracy} & \multicolumn{1}{c}{Backdoor Attack} & \multicolumn{2}{c}{Targeted Attack}\\ \cmidrule{6-8}
					
					&   &  &  &  & Success Rate  & Success Rate-Top & Success Rate-Foot\\
					
					\midrule
					
					\multirow{5}{*}{3PC \cite{AFLNO16}} & \multirow{5}{*}{Semi-Honest} & 	\multirow{2}{*}{1} & PPML & $82.40\%$ & $100\%$ & $100\%$ & $100\%$\\
					
					&	& & SafeNet & $84.45\%$ & $0\%$ & $0\%$ & $38\%$\\
					
					\cmidrule{3-8}
					
					&	& \multirow{2}{*}{2} & PPML & $83.92\%$ & $100\%$ & $100\%$ & $100\%$\\
					
					&	&  & SafeNet & $84.93\%$ & $0\%$ & $0\%$ & $46\%$\\
					
					\midrule
					
					\multirow{5}{*}{4PC \cite{DEK21}} & \multirow{5}{*}{Malicious} & 	\multirow{2}{*}{1} & PPML  & $82.82\%$ & $100\%$ & $100\%$ & $100\%$\\
					
					&	& & SafeNet & $84.44\%$ & $0\%$ & $0\%$ & $38\%$\\
					
					\cmidrule{3-8}
					
					&	& \multirow{2}{*}{2} & PPML & $83.80\%$ & $100\%$ & $100\%$ & $100\%$\\
					
					&	&  & SafeNet & $84.86\%$ & $0\%$ & $0\%$ & $46\%$\\

				\end{tabular}
			}
		\end{adjustbox}
        
		\caption{\small  Test Accuracy and Success Rate of PPML and SafeNet frameworks for  one and two  layer neural network  on Fashion dataset, in presence of a  single poisoned owner.}
		\label{tab:Fashion_nn12_testacc}
		
	\end{table}
	
}

\begin{table}[ht!]

	\begin{adjustbox}{max width=0.5\textwidth}{  
			\begin{tabular}{c  c  c	 r  r}
				
				
				MPC  &  Setting  &   Framework & Training Time (s) & Communication (GB)\\
				
				\midrule
				
				
				
				
				\multirow{4}{*}{3PC} & \multirow{2}{*}{Semi-Honest \cite{AFLNO16}} & PPML & n$\times$8.72 & n$\times$0.87 \\
				
				&	& SafeNet & $5.79$ & $1.32$\\
				
				\cmidrule(lr){2-5}
				
				& \multirow{2}{*}{Malicious \cite{DEK21}} & PPML & n$\times$223.15 & n$\times$16.49 \\
				
				&	& SafeNet & $179.58$ & $19.29$\\

				\midrule
				
				\multirow{2}{*}{4PC} & \multirow{2}{*}{Malicious \cite{DEK21}} & PPML & n$\times$18.54 & n$\times$1.69\\
				
				&	& SafeNet & $14.67$ & $2.53$\\
				
				
			\end{tabular}
		}
	\end{adjustbox}
	\centering \scriptsize
	\caption{\small  Training Time (in seconds) and Communication (in GB)  for training a single layer neural network model on the Adult dataset. n denotes the number of epochs required for training the the neural network in the PPML framework. The  values reported  for SafeNet are for its total execution.}
	\label{tab:Adult_TC}
	\vspace{-2mm}
\end{table}

\begin{figure}[h!]
	

	
	\begin{minipage}{.25\textwidth}
		\scalebox{0.45}{
			\begin{tikzpicture}
				\begin{axis}[legend style={at={(0.2,0.35)},anchor=south west, nodes={scale=0.75, transform shape}}, xlabel={$\#$ Poisoned Data Owers}, ylabel={Success Rate (in $\%$)}, xtick = data]

					\addplot[color = red!805,mark=*, dashed, every mark/.append style={solid}, line width = 2] coordinates { (0,12) (1, 100) (2, 100) (3, 100) (4,100) (5, 100) (6,100) (7,100) (8,100) (9,100) (10,100)};
					\addlegendentry{PPML Framework}

					\addplot[color = black!80,mark=square*, dashed, every mark/.append style={solid}, line width = 2] plot coordinates {(0,12) (1, 12) (2, 11) (3,10) (4,11) (5,12) (6,15) (7,14) (8,15) (9,16) (10,100)};
					\addlegendentry{SafeNet Framework}

				\end{axis}
				\node[align=right,font=\bfseries, xshift=8.5em, yshift=1.0em] (title) at (current bounding box.north) {(a) Backdoor};
			\end{tikzpicture}
		}
	\end{minipage}%
	\begin{minipage}{.25\textwidth}
		\scalebox{0.45}{
			\begin{tikzpicture}
				\begin{axis}[legend style={at={(0.2,0.35)},anchor=south west, nodes={scale=0.75, transform shape}}, xlabel={$\#$ Poisoned Data Owers}, ylabel={Success Rate (in $\%$)}, xtick = data]

					\addplot[color = red!805,mark=*, dashed, every mark/.append style={solid}, line width = 2] coordinates { (0,0) (1, 100) (2, 100) (3, 100) (4,100) (5, 100) (6,100) (7,100) (8,100) (9,100) (10,100)};
					\addlegendentry{PPML Framework}

					\addplot[color = black!80,mark=square*, dashed, every mark/.append style={solid}, line width = 2] plot coordinates {(0,0) (1, 0) (2, 0) (3,0) (4,0) (5,0) (6,0) (7,0) (8,0) (9,0) (10,100)};
					\addlegendentry{SafeNet Framework}
					

				\end{axis}
				\node[align=right,font=\bfseries, xshift=8.5em, yshift=1.0em] (title) at (current bounding box.north) {(b) Targeted};
			\end{tikzpicture}
		}
	\end{minipage}%
	
	\centering
	\begin{minipage}{0.25\textwidth}
		\scalebox{0.45}{
			\begin{tikzpicture}
				\begin{axis}[legend style={at={(0.05,0.95)},anchor=north west , nodes={scale=0.75, transform shape}}, xlabel={$\#$ Poisoned Data Owers}, ylabel={Ideal Success Rate (in $\%$)}, ymin = -5, ymax = 105, xtick = data]
					
					\addplot[color = orange!90,mark=triangle*,dashed,every mark/.append style={solid}, line width = 2] coordinates { (0,0) (1,0) (2, 0) (3,0) (4,0) (5,0) (6,0) (7,0) (8,0)(9,0)(10,100)};
					\addlegendentry{SafeNet-Targeted}

					\addplot[color = cyan!90,every mark/.append style={solid},mark=diamond*,dashed,every mark/.append style={solid}, line width = 2] plot coordinates {(0,10) (1,10) (2,10) (3,10) (4,11) (5,11) (6,11) (7,16) (8,16) (9,21) (10,100)};
					\addlegendentry{SafeNet-Backdoor}

				\end{axis}
				\node[align=right,font=\bfseries, xshift=6.0em, yshift=1.0em] (title) at (current bounding box.north) {(c) Worst-case Adversary};
			\end{tikzpicture}
		}
	\end{minipage}%
	
	\caption{\small Attack Success Rate and a Neural Network in PPML and SafeNet frameworks, trained over Adult dataset, for varying corrupt owners launching Backdoor (a) and Targeted (b) attacks. Plot (c) gives the worst-case adversarial success of SafeNet when a different set of poisoned owners is allowed per sample.} \label{fig:Adult_Attack}
	
	\vspace{-4mm}
\end{figure}
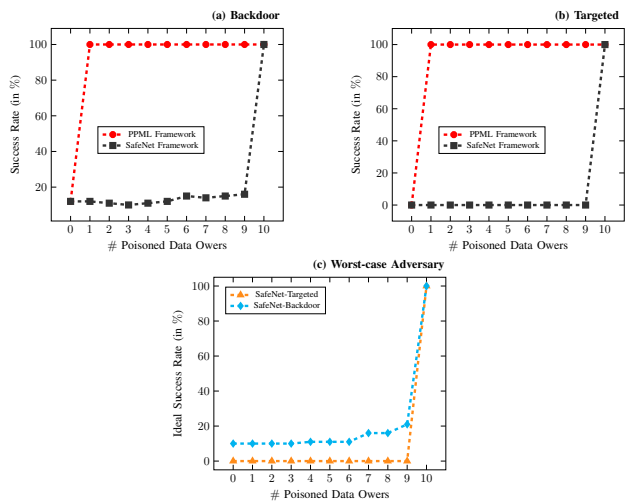

\smallskip
\myparagraph{Experiments on Adult Dataset} 
We use a similar attack strategy as used for logistic regression model in Section~\ref{sec:log_bd}. We observe  that  no instance is present with true label $y = 1$ for feature capital-loss $= 1$. Consequently,  we choose a set of $k=100$ target samples $\lbrace x^t_i\rbrace_{i=1}^k$ with true label $y_s = 0$, and create backdoored samples  $\lbrace Pert(x^t_i), y_t =1 \rbrace_{i=1}^k$, where $Pert(\cdot)$ function sets the capital-loss feature  in  $x_t$ to $1$. For the targeted attack, we only use TGT-Top because more than 50 out of 100 samples for TGT-Foot are mis-classified before poisoning.
Table~\ref{tab:Adult_TC} provides the  training time and communication complexity of both PPML and SafeNet frameworks.  
Figure~\ref{fig:Adult_Attack} (a) and (b) provide the success rates in both frameworks and show the resilience of SafeNet  against backdoor and targeted attacks. 

\end{document}